\documentclass[11pt]{article}
\usepackage{amsmath,amssymb,amsfonts} 
\usepackage{epsfig} \usepackage{latexsym,nicefrac,bbm}
\usepackage{xspace}
\usepackage{color,fancybox,graphicx,subfigure,fullpage}
\usepackage[top=1in, bottom=1in, left=1in, right=1in]{geometry}
\usepackage{tabularx}
\usepackage{hyperref} 
\usepackage{pdfsync}
\usepackage[boxruled]{algorithm2e}
\usepackage{multicol}
\usepackage{enumitem}
\usepackage{multirow}
\usepackage{url}

\renewcommand{\epsilon}{\varepsilon}

\newcommand{\eps}{\varepsilon}

\newcommand{\poly}{\mathrm{poly}}

\newcommand\N{\mathbb N}
\newcommand\R{\mathbb R}

\newtheorem{theorem}{Theorem}[section]

\newtheorem{definition}{Definition}[section]
\newtheorem{lemma}[theorem]{Lemma}
\newtheorem{remark}[theorem]{Remark}
\newtheorem{proposition}[theorem]{Proposition}
\newtheorem{corollary}[theorem]{Corollary}

\newtheorem{problem}[theorem]{Problem}

\newenvironment{proof}{\begin{trivlist} \item {\bf Proof:~~}}
  {\qed\end{trivlist}}

\def\FullBox{\hbox{\vrule width 6pt height 6pt depth 0pt}}

\def\qed{\ifmmode\qquad\FullBox\else{\unskip\nobreak\hfil
\penalty50\hskip1em\null\nobreak\hfil\FullBox
\parfillskip=0pt\finalhyphendemerits=0\endgraf}\fi}

\newcommand{\C}{\mathbb{C}}
\renewcommand{\R}{\mathbb{R}}
\renewcommand{\N}{\mathbb{N}}

\DeclareMathOperator{\Tr}{Tr}

\DeclareMathOperator{\diag}{diag}

\DeclareMathOperator{\type}{type}

\def\U{{\rm U}}

\usepackage[T1]{fontenc}
\usepackage{lmodern}

\title{Sampling Matrices from
Harish-Chandra--Itzykson--Zuber Densities with Applications to Quantum Inference and Differential Privacy}
\author{Jonathan Leake \\ TU Berlin \and  Colin S. McSwiggen \\ University of Tokyo \and Nisheeth K. Vishnoi \\ Yale University}

\begin{document}

\maketitle

\sloppy

\begin{abstract}
Given two $n \times n$ Hermitian matrices $Y$ and $\Lambda$, the Harish-Chandra--Itzykson--Zuber (HCIZ) distribution on the unitary group $\mathrm{U}(n)$ is $e^{\Tr(U \Lambda U^*Y)}d\mu(U)$, where $\mu$ is the Haar measure on $\mathrm{U}(n)$.
The density $e^{\Tr(U \Lambda U^*Y)}$ is known as the HCIZ density.
Random unitary matrices distributed according to the HCIZ density are important in various settings in physics and random matrix theory.
However, the basic question of how to sample efficiently from the HCIZ distribution has remained open.
We present two efficient algorithms to sample matrices from distributions that are close to the HCIZ distribution. 
The first algorithm outputs samples that are $\xi$-close in the total variation distance  
and requires at most polynomially many arithmetic operations in $\log 1/\xi$ and the number of bits needed to encode $Y$ and $\Lambda$.
The second algorithm comes with a stronger guarantee that the samples are $\xi$-close in infinity divergence, however the number of arithmetic operations depends polynomially on $1/\xi$, the number of bits needed to encode $Y$ and $\Lambda$, and the differences of the largest and the smallest eigenvalues of $Y$ and $\Lambda$.

HCIZ densities can also be viewed as exponential densities on $\mathrm{U}(n)$-orbits, and in this setting, they have been studied implicitly or explicitly in statistics, machine learning, and theoretical computer science.
Thus, our results have the following  applications:  1) an efficient algorithm to sample from complex versions of matrix Langevin distributions studied in statistics \cite{ChikusePaper,ChikuseBook}, 2) an efficient algorithm to sample from continuous maximum entropy distributions over unitary orbits \cite{LeakeV20,LeakeV20b}, which in turn implies an efficient algorithm to sample a pure quantum state from the entropy-maximizing ensemble representing a given density matrix, and 3) an efficient algorithm for differentially private rank-$k$ approximation \cite{Kamalika,KTalwar} that comes with improved utility bounds for $k>1$.

\end{abstract}

\thispagestyle{empty}
\newpage

\thispagestyle{empty}

\tableofcontents
\newpage

\setcounter{page}{1}

\section{Introduction}

Let $\U(n)$ denote the group of $n \times n$ unitary matrices and let $\mu$ denote the Haar probability measure on $\U(n)$.
Given $n \times n$ Hermitian matrices $Y$ and $\Lambda$, consider the following measure on $\U(n)$:
  \begin{equation} \label{eqn:hciz-d}
      e^{\Tr(U \Lambda U^*Y)} d\mu(U).
 \end{equation}
The corresponding density is  referred to as the Harish-Chandra--Itzykson--Zuber (HCIZ) density 
and has been extensively studied, implicitly and explicitly, in physics, random matrix theory, statistics, and theoretical computer science.
A major result  about the HCIZ density is that its integral over $\mathrm{U}(n)$ admits an exact expression as a determinant.
\begin{theorem}[HCIZ integral formula]\label{thm:HCIZ}
    For $n \times n$ Hermitian matrices $Y$ and $\Lambda$ with distinct eigenvalues $y_1 > \cdots > y_n$ and $\lambda_1 > \cdots > \lambda_n$ respectively, we have the following:\footnote{Although (\ref{eqn:hciz}) assumes that all $y_i$ and $\lambda_i$ are distinct, when this is not the case one can still use Theorem \ref{thm:HCIZ} to obtain an exact determinantal formula for the HCIZ integral, simply by applying L'H\^opital's rule to the right-hand side of (\ref{eqn:hciz}).
}
    \begin{equation} \label{eqn:hciz}
        \int_{\U(n)} e^{\Tr(U \Lambda U^*Y)} d\mu(U) = \left(\prod_{p=1}^{n-1} p!\right) \frac{\det([e^{y_i \lambda_j}]_{1 \leq i,j \leq n})}{\prod_{i < j} (y_i - y_j)(\lambda_i - \lambda_j)}.
    \end{equation}
\end{theorem}

\noindent
Theorem \ref{thm:HCIZ} was proved by Harish-Chandra \cite{HarishChandra1957}  
and by Itzykson and Zuber \cite{IZ}.
See the post by Terry Tao \cite{Tao} and the notes of the second author \cite{McS-expository} for  more on the HCIZ integral.

\smallskip
\noindent
{\bf Physics and random matrix theory.} Matrices distributed according to the HCIZ density are important in various settings in physics and random matrix theory. 
For instance, they appear in multi-matrix models in quantum field theory and string theory \cite{IZ,DGZ}, and they are also related to models of coupled Gaussian matrices \cite{IZ} that have been used to solve the Ising model on a planar random lattice \cite{Kaz, BouKaz}.
In particular, the moments of HCIZ distributed unitary matrices play a role in computing correlation functions for matrix models of gauge theories and have been studied extensively since the 1990's \cite{Morozov, Sha, Eynard-note, E-PF, PEDZ}.  

The HCIZ integral also arises in many other places in random matrix theory.  Notably, it occurs in expressions for the joint spectral densities of a number of matrix ensembles, such as Wishart matrices and off-center Wigner matrices \cite{AGsurvey}.  
However, the basic question about sampling from the HCIZ distribution has remained open.

The problem of sampling from the HCIZ distribution can be equivalently cast as the problem of sampling according to an exponential density specified by $Y$ on the $\mathrm{U}(n)$-orbit of $\Lambda$.
Let
$$\mathcal{O}_\Lambda := \{ U \Lambda U^* \ | \ U \in \U(n) \}$$ denote the orbit of $\Lambda$ under the conjugation action of $\U(n)$.
Every such orbit contains a diagonal element $\diag(\lambda)$ where $\lambda$ is the sequence of eigenvalues of $\Lambda$ listed in non-increasing order, and thus we always assume $\Lambda = \diag(\lambda)$.
Further, one can write any $X \in \mathcal{O}_\Lambda$ as $X = U \Lambda U^*$ for some $U \in \U(n)$, so that the density of the measure \eqref{eqn:hciz-d} can be rewritten as
   $$  e^{\langle Y, X \rangle} =  e^{\Tr(U \Lambda U^*Y)},$$ where $\langle Y, X \rangle := \Tr (Y^*X).$ 
Thus, we arrive at the following  problem.

\vspace{-2mm}
\begin{problem}[\bf Sampling from unitary orbits]\label{prb:sampling} Given two $n \times n$ Hermitian matrices $Y$ and $\Lambda = \diag(\lambda)$, sample an $X \in \mathcal{O}_\Lambda$ from the probability distribution
\begin{equation}\label{eq:HCIZ_density}
      d\nu(X) \propto e^{\langle Y, X \rangle} d\mu_\Lambda(X),
    \end{equation}
where $\mu_\Lambda$ is the $\U(n)$-invariant probability measure on $\mathcal{O}_\Lambda$. 
\end{problem}
\vspace{-2mm}

\noindent
{\bf Statistics.} Distributions of the kind mentioned in  \eqref{eq:HCIZ_density} have also been studied under the name \emph{matrix Langevin} or {\em matrix Bingham} in statistics   \cite{ChikusePaper,ChikuseBook}.
The difference is that these  distributions are supported on orbits of the orthogonal group rather than the unitary group. 
Obtaining efficient algorithms to sample from such distributions 
 is left as an open problem; see Section 2.5.2 of \cite{ChikuseBook}.

\paragraph{Continuous maximum entropy distributions over matrix manifolds.} In recent works of \cite{LeakeV20,LeakeV20b}, distributions as in \eqref{eq:HCIZ_density} arose as solutions to maximum entropy problems over manifolds, with applications to computing the entropy-maximizing representation of a quantum density matrix as an ensemble of ``pure states''.
Concretely,  the authors study the following problem: given a matrix $A$ in the convex hull of $ \mathcal{O}_\Lambda$, compute  the probability density supported on $\mathcal{O}_\Lambda$ whose expected value
is $A$ and that minimizes the Kullback-Leibler divergence to $\mu_\Lambda$.
As an example, if we let $\Lambda$ be the diagonal matrix with exactly one $1$ and rest $0$s, 
the convex hull of $\mathcal{O}_\Lambda$ is exactly the set of PSD matrices with trace one -- density matrices. 
Thus, in this case, the solution to the above entropy problem gives a way to ``infer'' an ensemble of pure states corresponding to a given density matrix $A$, following the principle of maximum entropy \cite{Jaynes1,Jaynes2,Band1976,SlaterEntropy1}.
The authors show that the solution to the above optimization problem gives rise to the distribution of the form $e^{\langle Y^\star,X\rangle} d\mu_\Lambda(X)$ for some $Y^\star.$ 
Their main result is a polynomial-time algorithm to find this optimal solution that runs in time, roughly, the number of bits needed to represent $A$ and the distance of $A$ to the ``boundary'' of the convex hull of $\mathcal{O}_\Lambda$.  
While \cite{LeakeV20, LeakeV20b} gave polynomial-time algorithms to compute the optimal value $Y^\star$, designing an algorithm to sample from the corresponding distribution
 was left as an open problem.

\paragraph{Differentially private algorithms for low-rank approximation.} An important technique to obtain differentially private algorithms is the exponential mechanism due to McSherry and Talwar \cite{MTalwar}.
In the context of rank-$k$ approximation of a given matrix, it amounts  to sampling from an exponential density of the type \eqref{eq:HCIZ_density} on the orbit corresponding to rank-$k$ projections \cite{Kamalika,KTalwar}. 
To the best of our knowledge, the only  result on this problem is an approximate algorithm for the rank-$1$ case given by \cite{KTalwar}. 
They left it an open problem to simplify their rank-$1$ algorithm and also come up with an algorithm to sample from the corresponding exponential mechanism for the  rank-$k$ case when $k>1$.
%

\paragraph{Our contributions.} We present two efficient algorithms to approximately sample from HCIZ distributions or, equivalently, from exponential densities on unitary orbits.
The two algorithms differ in how they approximate the target HCIZ distribution.
Our first algorithm approximates in the total variation distance and is sufficient for many applications.

\begin{theorem}[\bf Main result -- total variation distance]\label{thm:main}
There is an algorithm that, given a $\xi>0$ and $n \times n$ Hermitian matrices $\Lambda = \diag(\lambda)$ and $Y = \diag(y)$, outputs a matrix $X$ that is distributed according to a distribution that is $\xi$-close in TV distance to 
$
d\nu(X) \propto e^{\langle Y, X \rangle} d\mu_\Lambda(X).
$
The number of arithmetic operations required to run the algorithm is polynomial in $\log\frac{1}{\xi}$ and the number bits required to represent $y$ and $\lambda$.
\end{theorem}
%
%
\noindent
Our second algorithm approximates the target distribution in the (stronger) infinity divergence (Definition \ref{def:distr_dist}) and is used for the differential privacy application.
\begin{theorem}[\bf Main result -- infinity divergence]\label{thm:main2}
There is an algorithm that, given a $\xi>0$ and $n \times n$ Hermitian matrices $\Lambda = \diag(\lambda)$ and $Y = \diag(y)$, outputs a matrix $X$ that is distributed according to a distribution that is $\xi$-close in infinity divergence distance to 
$
d\nu(X) \propto e^{\langle Y, X \rangle} d\mu_\Lambda(X).
$
The number of arithmetic operations required to run the algorithm is polynomial in  $\lambda_{\mathrm{max}} - \lambda_{\mathrm{min}}$, $y_{\mathrm{max}} - y_{\mathrm{min}}$, $\frac{1}{\xi}$, and the number bits required to represent $y$ and $\lambda$.
\end{theorem}
Note that while the approximation guarantee of the second algorithm (Theorem \ref{thm:main2}) is better, the number of arithmetic operations it performs depends polynomially on $\frac{1}{\xi}$ as opposed to poly-logarithmically in $\frac{1}{\xi}$ as in the first algorithm (Theorem \ref{thm:main}).
We leave it as an open problem to give an algorithm which samples from a distribution that is $\xi$-close in infinity divergence but whose arithmetic operations depend polynomially on $\log \frac{1}{\xi}$.

Theorem \ref{thm:main} enables efficient numerical simulation of models in physics  and random matrix theory where HCIZ densities arise.
Our algorithms also make progress on the open problems mentioned earlier. 
In particular, Theorem \ref{thm:main} immediately gives  efficient algorithms to 
\begin{itemize}
    \item  sample from complex matrix Langevin distributions \cite{ChikusePaper,ChikuseBook} and
    \item  sample from continuous maximum entropy distributions on unitary orbits studied by \cite{LeakeV20,LeakeV20b}, implying an efficient algorithm to sample a pure quantum state from the entropy-maximizing ensemble corresponding to a given density matrix.
    \end{itemize}
\noindent
Moreover, Theorem \ref{thm:main2} implies an efficient algorithm for the exponential mechanism for differentially private rank-$k$ approximation  \cite{MTalwar,Kamalika,KTalwar}.
As a consequence, 
we show that Theorem \ref{thm:main2} allows us to obtain an efficient differentially private rank-$k$ approximation with improved utility.

\begin{theorem}[\bf Differentially private low-rank approximation]\label{thm:diff-main}
    There is a randomized algorithm  that, given a $d \times d$ positive semidefinite matrix $A$ and its eigenvalues $\gamma_1 \geq \cdots \geq \gamma_d > 0$, an integer $1 \leq k \leq d$, an $\eps>0$ and a small $\delta>0$, outputs a rank-$k$ $d \times d$ Hermitian projection $P$ such that 
    $
        \mathbb{E}_P\left[\langle A, P \rangle\right] \geq (1-\delta) \sum_{i=1}^k \gamma_i
    $
    as long as $\sum_{i=1}^k \gamma_i \geq C \cdot \frac{dk}{\epsilon \delta} \cdot \log\frac{1}{\delta}$, where $C >0$ is a universal constant.
    This algorithm is $\eps$-differentially private and requires a number of arithmetic operations polynomial in $\gamma_1-\gamma_d$, $\frac{1}{\epsilon}$, and the number of bits required to represent $\gamma$.
\end{theorem}

\noindent
Note that the utility bound promised by Theorem \ref{thm:diff-main} is about $\frac{dk}{\epsilon}$ when compared to the utility bound of roughly $\frac{dk^3}{\eps}$ due to \cite{KTalwar} for $k>1$.

The proofs of Theorems \ref{thm:main} and \ref{thm:main2} are identical 
except for an intermediate step that we mention below.
One of the key difficulties in sampling from an HCIZ distribution is that its domain, a unitary orbit, is a non-convex algebraic manifold.
The individual entries of the desired sample matrix are highly correlated due to the algebraic constraints that define the orbit, which makes it difficult to break the problem into lower-dimensional subproblems.
Our main technical contribution is to reduce the problem of sampling from an exponential density on a unitary orbit to sampling from an exponential density on a bounded convex polytope.
In particular, we use an alternative parameterization of unitary orbits based on the {\it Rayleigh map}, which sends a Hermitian matrix $X$ to a natural organization of the eigenvalues of all leading principal submatrices of $X$.
The image of each $\U(n)$-orbit under the Rayleigh map is a convex polytope called a {\it Gelfand--Tsetlin (GT) polytope}, which is cut out by
linear inequalities given by the interlacing properties of the eigenvalues.
This mapping reveals a recursive structure intrinsic to $\U(n)$-orbits, which is hard to see directly in the ambient space of matrices.
The Rayleigh map from a given $\U(n)$-orbit to the corresponding GT polytope is not injective.
However, one can show that 1) the HCIZ density on the orbit pushes forward to an exponential density on the polytope, and 2) the HCIZ density is constant on the fibers of the Rayleigh map.
Therefore, to solve the sampling problem on the orbit, it suffices to sample a point from an exponential density on the GT polytope, and then sample a Hermitian matrix uniformly at random from the fiber of the Rayleigh map over that point.
To sample from the GT polytope, we use results of \cite{LovaszVempala06} for Theorem \ref{thm:main} and results of \cite{BassilySmithThakurta14} for Theorem \ref{thm:main2}.

We give a detailed technical overview of the algorithms and the proofs of Theorems \ref{thm:main}, \ref{thm:main2}, and \ref{thm:diff-main} in Section \ref{sec:technical-overview}.
The formal algorithms and proofs appear in Sections \ref{sec:sampling_algo} and \ref{app:diff-priv}.

\section{Rayleigh triangles, Gelfand--Tsetlin polytopes, and unitary orbits}

In this section we introduce some definitions and facts that we will need in what follows.  
In particular, we discuss two types of combinatorial objects that are fundamental to the geometry of Hermitian matrices: Rayleigh triangles and Gelfand--Tsetlin polytopes.

\begin{definition}[\bf Rayleigh triangle]\label{def:rayleigh}
For an integer $n \geq 1$, a {\it Rayleigh triangle} is a triangular array of real numbers $R = (R_{i,j})_{1 \le i \le j \le n}$ satisfying the {\it interlacing relations}
 
\begin{equation} 
R_{i,j} \ge R_{i,j-1} \ge R_{i+1,j} \qquad \textrm{for all } 1 \le i < j \le n.
\end{equation}
The vector $R_{\bullet,j} = (R_{1,j},\, \hdots,\, R_{j,j}) \in \R^j$ is called the $j$th row, and $R_{\bullet,n} \in \R^n$ is called the {\bf top} row.  
If we fix $R_{\bullet,n}$, we can regard the numbers $R_{i,j}$, $j \le n-1$ as coordinates of a point in $\R^{n(n-1)/2}$.  
\end{definition}

\noindent
Note that the indexing for the Rayleigh triangle is different from that of matrix notation: the top row is indexed by $n$.

\begin{definition}[\bf Gelfand--Tsetlin polytope]\label{def:GT}
Given a vector $\lambda \in \R^n$ with $\lambda_1 \ge \cdots \ge \lambda_n$, the {\it Gelfand--Tsetlin polytope} $GT(\lambda)$ is the convex polytope in $\R^{n(n-1)/2}$ consisting of all Rayleigh triangles with top row equal to $\lambda$. 
\end{definition}

\noindent
Thus, $GT(\lambda)$ is the polytope cut out by the interlacing inequalities (\ref{eqn:interlacing-relns}), with $R_{\bullet,n} = \lambda$ fixed.
In other words, the following system of $n$ equalities and $n \times (n-1)$ inequalities determines $GT(\lambda)$:
\begin{equation}\label{eq:toprow}
    R_{i,n} = \lambda_i \qquad \textrm{for all } 1 \le i \le n,
\end{equation}
and
\begin{equation} \label{eqn:interlacing-relns}
R_{i,j} - R_{i,j-1} \geq 0  \qquad \mbox{and} \qquad  R_{i,j-1}- R_{i+1,j} \geq 0  \qquad \textrm{for all } 1 \le i < j \le n.
\end{equation}
Note that if all entries of $\lambda$ are distinct, then $GT(\lambda)$ is full-dimensional in $\R^{n(n-1)/2}$ (with coordinates $R_{i,j}$ for $1 \leq i \leq j \leq n-1$), and every inequality given by (\ref{eqn:interlacing-relns}) above is essential.
However if entries of $\lambda$ coincide, then some of the inequalities of (\ref{eqn:interlacing-relns}) become equalities, and $GT(\lambda)$ lies in some affine subspace of $\R^{n(n-1)/2}$.
In particular, if $\lambda_p = \lambda_{p+1} = \cdots = \lambda_q$ then for any $R \in GT(\lambda)$ we have that $R_{i,j} = \lambda_p$ for all $i,j$ such that $p \leq i \leq q+j-n$.
On the other hand, every inequality of (\ref{eqn:interlacing-relns}) not associated to such a fixed entry $R_{i,j}$ is essential.
Using this observation, for any fixed $\lambda$ it is straightforward to determine the affine subspace in which $GT(\lambda)$ has non-empty interior.

The following is then a corollary of a classical result of linear algebra known as the Cauchy--Rayleigh interlacing theorem.
\begin{proposition}[\bf Hermitian matrices, interlacing, and Rayleigh triangles]\label{prop:Herm}
Given an $n \times n$ Hermitian matrix $X$, denote by $X[k]$ its $k$th leading principal submatrix (that is, the $k \times k$ submatrix in the upper left corner of $X$).  
Let $\lambda_{1,k} \geq \cdots \geq \lambda_{k,k}$ be the eigenvalues (which are real) of $X[k]$.
Then the eigenvalues $(\lambda_{j,k})_{1 \le j \le k \le n}$ of the leading submatrices of $X$ form a Rayleigh triangle, which we write as $\mathcal{R}(X)$.
\end{proposition}

\begin{definition}[\bf Type vector]\label{def:type}
The {\it type vector} of $R$ is defined by
$$
 \mathrm{type}(R) = \Big( R_{1,1},\, R_{1,2} + R_{2,2} - R_{1,1},\, \hdots,\, \sum_{i=1}^n R_{i,n} - \sum_{j=1}^{n-1} R_{j,n-1} \Big).$$
If $R = \mathcal{R}(X)$ for some Hermitian $X$, then $\type(R) = (X_{11}, \hdots, X_{nn})$ is the diagonal of $X$.
\end{definition}

\begin{definition}[\bf Orbits of $\U(n)$] Given a vector $\lambda \in \R^n$ as above, write $\Lambda = \diag(\lambda)$ and let $\mathcal{O}_\Lambda = \{ U \Lambda U^* \ | \ U \in \U (n) \}$ be the unitary conjugation orbit of $\Lambda$.  
Let $\mu_\Lambda$ be the uniform probability measure on $\mathcal{O}_\Lambda$, i.e., the unique probability measure on $\mathcal{O}_\Lambda$ that is invariant under the conjugation action of $\U(n)$.
\end{definition}
\noindent
It can be shown that the image $\mathcal{R}(\mathcal{O}_\Lambda)$ is $GT(\lambda)$.
In fact, the following stronger result is true: the uniform measure on $\mathcal{O}_\Lambda$ maps to the uniform measure on $GT(\lambda)$; see e.g.~\cite{Barysh, Neretin, NO, JB-minors}.

\begin{proposition}[\bf Pushforward of Haar measure] \label{prop:R-pushfwd}
The pushforward of the Haar measure $\mu_\Lambda$ on $\mathcal{O}_\Lambda$ by the map $\mathcal{R}$ is the uniform probability measure on $GT(\lambda)$.
\end{proposition}

\noindent
Note that in the above result,
the pushforward distribution is a restriction of the
Lebesgue measure on the {\it affine span} of $GT(\lambda)$, which is the minimal affine subspace of $\R^{n(n-1)/2}$ containing $GT(\lambda)$.
This distinction only matters when not all $\lambda_i$ are distinct, since in this case $GT(\lambda)$ has dimension less than $n(n-1)/2$ so that its volume in the ambient space $\R^{n(n-1)/2}$ is zero.

Proposition \ref{prop:R-pushfwd} allows us to prove the following crucial fact that the image of an exponential density on a unitary orbit is an exponential density on the GT polytope.

\begin{theorem}[Pushforward of the HCIZ density]\label{thm:equivalence}
If $Y=\diag(y)$ for some real vector $y$, then the pushforward of the measure $e^{\langle Y, X \rangle} d \mu_\Lambda(X)$ through the map $\mathcal{R}$ is
$$
 \mathcal{R}_* \big[ e^{\langle Y, X \rangle} d \mu_\Lambda(X) \big] = \mathrm{Vol}(GT(\lambda))^{-1}  e^{ \langle y, \mathrm{type}(P)\rangle } dP,
$$
where $dP$ denotes the Lebesgue measure on the affine span of $GT(\lambda)$.
\end{theorem}

\begin{proof}
First note that  $(X_{11}, \hdots, X_{nn}) = \mathrm{type}(\mathcal{R}(X))$, as mentioned in Definition \ref{def:type}.
Thus, we have
$$\langle Y, X \rangle = \sum_{i=1}^n y_i X_{ii} = \langle y,  \mathrm{type}(\mathcal{R}(X))\rangle.$$
The result then follows from the fact that the pushforward of $\mu_\Lambda$ through $\mathcal{R}$ is uniform on $GT(\lambda)$ by Proposition \ref{prop:R-pushfwd}.
\end{proof}

\noindent
Finally, to prove the correctness of our  sampling algorithm, we will need to describe the set of Hermitian matrices that map to a given Rayleigh triangle under $\mathcal{R}$.

\begin{definition}[\bf Fiber over a Rayleigh triangle]\label{def:fiber}
Given $R \in GT(\lambda)$, the fiber of the map $\mathcal{R}$ over $R$ is the set
$\mathcal{R}^{-1}(R) = \{ X \in \mathcal{O}_\Lambda \ | \ \mathcal{R}(X) = R \}.$
\end{definition}

\noindent
The fiber $\mathcal{R}^{-1}(R)$ is a compact subset of $\mathcal{O}_\Lambda$. The {\it uniform probability measure} on $\mathcal{R}^{-1}(R)$ is characterized by the property that if $X$ is uniformly distributed in $\mathcal{R}^{-1}(R)$, then for $1 < k \le n$, $X[k]$ is uniformly distributed on the compact manifold $\mathcal{H}(X[k-1]; R_{\bullet,k})$ of $k \times k$ Hermitian matrices with eigenvalues $R_{\bullet,k}$ and leading $(k-1) \times (k-1)$ submatrix equal to $X[k-1]$. We will show below in Lemma \ref{lem:samp_fiber_correctness} that $\mathcal{H}(X[k-1]; R_{\bullet,k})$ is a product of spheres.

Before moving on, we formally define the notions of distance between distributions which are relevant to our main results.

\begin{definition}[\bf Notions of distance between distributions]\label{def:distr_dist}
    Given two distributions (i.e.~Borel probability measures) $\mu,\nu$ on $\mathcal{S} \subseteq \R^n$, we define the \emph{total variation distance} between $\mu$ and $\nu$ as
    $
        \|\mu - \nu\|_{\mathrm{TV}} := \sup_{S \subset \mathcal{S}} |\mu(S) - \nu(S)|.
    $
    When $\mu,\nu$ have continuous density functions $f,g$ respectively (with respect to the same base measure) on $\mathcal{S}$, we further define the \emph{infinity divergence} between $\mu$ and $\nu$ as
    $
        D_\infty(\mu \| \nu) := \log \sup_{x \in \mathcal{S}} \frac{f(x)}{g(x)}.
    $
\end{definition}

\section{Technical overview} \label{sec:technical-overview}

In this section, we give an overview of the algorithms and the proof of our main results (Theorems \ref{thm:main} and \ref{thm:main2}), leaving the full details to Section \ref{sec:sampling_algo}.
We also give an overview of the proof of the differential privacy result (Theorem \ref{thm:diff-main}), leaving the full details to Section \ref{app:diff-priv}.
Here we will emphasize the important ideas and and concepts from the various parts of the proof without going into too much detail.
For the interested reader, we will provide links to other relevant sections of the full proof throughout this overview.

Throughout, $\Lambda$ and $Y$ will always be $n \times n$ real diagonal matrices.
The unitary orbit $\mathcal{O}_\Lambda$ is defined as the set of all matrices $U \Lambda U^*$ obtained by conjugating $\Lambda$ by any unitary matrix.
The measure $d\mu_\Lambda(X)$ is the unitarily-invariant probability measure on $\mathcal{O}_\Lambda$, which means that
\[
    \int_{\mathcal{O}_\Lambda} f(X)\, d\mu_\Lambda(X) = \int_{\mathcal{O}_\Lambda} f(UXU^*)\, d\mu_\Lambda(X)
\]
for all integrable functions $f$ and all unitary matrices $U \in \U(n)$.
The goal of our algorithms is to return a sample $X$ from the unitary orbit $\mathcal{O}_\Lambda$, such that the distribution of $X$ is proportional to $e^{\langle Y, X \rangle} d\mu_\Lambda(X)$.

\subsection{The uniform case} \label{sec:uniform_case}

Let us first consider a simple case, when $Y = 0$.
In this case the distribution we want to sample from is precisely the unitarily invariant (uniform) distribution on $\mathcal{O}_\Lambda$.
Unitary invariance implies that sampling $X$ from this distribution on $\mathcal{O}_\Lambda$ is equivalent to sampling $U$ from the Haar probability measure on $\U(n)$ and taking $U \Lambda U^*$ as our sample in $\mathcal{O}_\Lambda$.
Sampling $U$ from the Haar probability measure on $\U(n)$ then has a classical solution:
Inductively sample orthogonal unit vectors $u_1,u_2,\ldots,u_n$ from $\C^n$ by projecting and normalizing random Gaussian vectors.
We can then construct a random matrix $U \in \U(n)$ by setting $u_1,u_2,\ldots,u_n$ as the columns of $U$.
This shows that sampling from $\mathcal{O}_\Lambda$ in the case of $Y=0$ has a simple, intuitive solution.

\paragraph{Difficulty in extending the algorithm for the uniform case.}
For general $Y$ however, the situation quickly becomes more complicated.
The first observation is that unitary invariance is immediately lost, since generically we have
$$
    e^{\langle Y, X \rangle} \neq e^{\langle Y, UXU^* \rangle}.
$$
This means the method used for $Y=0$ breaks down, as there is no clear way to generalize the above simple algorithm to exponential weightings of the Haar measure.
This is even true in the most basic case when $\mathcal{O}_\Lambda$ is the set of rank-one projections (when $\Lambda = \diag(1,0,\ldots,0)$), and the difficulty in this case was already realized in several previous works \cite{ChikuseBook,KTalwar,LeakeV20}.

Even though the density is not unitarily invariant, there is still significant symmetry coming from the structure of the orbit $\mathcal{O}_\Lambda$.
This symmetry leads to the HCIZ integral formula (Theorem \ref{thm:HCIZ}), which gives an efficiently computable formula for the partition function of the HCIZ density.
(It should be noted that the proof of this formula is highly non-trivial: Harish-Chandra's original proof from 1957 can be viewed as a starting point for much of the modern theory of quantum integrable systems \cite{HarishChandra1957, McS-expository}.)
Typically, such an explicit formula for the partition function can be translated into an algorithm for sampling, but it is not clear how to do this for the unitary orbit $\mathcal{O}_\Lambda$.

\subsection{Searching for self-reducibility}

In the world of discrete distributions, the seminal work of \cite{JerrumVV86} gives a general way to sample from a distribution using an oracle for the associated partition function.
The key property needed to utilize their results is that the distribution needs to be \emph{self-reducible}.
A problem is said to be self-reducible if, roughly speaking, a problem instance with input size $n$ can be efficiently reduced to another instance of the same problem with input size $n - 1$.

As an example of where the ability to compute the partition function can lead to an efficient sampling algorithm, consider the case of sampling matchings from a graph.
To uniformly sample a matching, one can first choose an edge $e$ in the graph and then compute the number $k_e$ of matchings that contain $e$, and the number $l_e$ of matchings that do not contain $e$.
The edge $e$ is then included in the output matching with probability $k_e / l_e$, and the original problem can be reduced to finding a perfect matching in the smaller graph obtained by removing the vertices joined by $e$.
To sample non-uniform matchings, the values of $k_e$ and $l_e$ are replaced by evaluations of the partition function.

In our world of continuous distributions on unitary orbits, it is not obvious how to perform a self-reduction similar to that of the discrete world,
even though we have a formula for the partition function.
The obstacle is that self-reducibility depends on preserving the original problem structure: we must reduce to an instance of the same problem, but with smaller input.

One approach towards this is to iteratively sample the individual entries or columns of the matrix, and then to interpret the remaining entries of the matrix as a smaller instance of the original problem conditioned on the previously selected entries.
The issue with this approach is that the entries of a matrix $X$ in the unitary orbit $\mathcal{O}_\Lambda$ are highly correlated due to the algebraic constraint that $X = U \Lambda U^*$ for some $UU^* = I$.
This means that the problem of sampling from a given distribution on the orbit conditional on one or more matrix entries is a priori very different from the original problem, and much more complicated.

There is an alternative way to view a matrix $X \in \mathcal{O}_\Lambda$: in terms of its eigenvalues.
The eigenvalues and eigenvectors together determine the matrix $X$ completely.
Further, one can (almost) recover the eigenvectors of $X$ using the eigenvalues of the principal submatrices of $X$ (see \cite{Tao_eigen}).
When the eigenvalues are distinct, one finds
\[
    |v_{i,j}|^2 = \frac{\prod_{k=1}^{n-1} (\lambda_i(X) - \lambda_k(X_j))}{\prod_{k \neq i} (\lambda_i(X) - \lambda_k(X))},
\]
where $\lambda_i(X)$ is the $i^\text{th}$ largest eignevalue of $X$, $X_k$ is the principal submatrix of $X$ with the $k^\text{th}$ row and column removed, and $v_i$ is the eigenvector corresponding to $\lambda_i(X)$.
That is, with the extra information of the eigenvalues of the principal submatrices of $X$, one can determine the eigenvectors of $X$ up to the (complex) sign of the entries.

This is a good sign for us, as it hints at some inductive structure in the eigenvalues of $X$.
Can we now understand this relationship between the eigenvectors and eigenvalues of principal submatrices in some recursive manner?
As a matter of fact, by considering the matrix $X$ in terms of all of its \emph{leading} prinicpal submatrices, we are able to prove a similar result (see Section \ref{sec:algo_step2}).
And not only that, but it turns out that this eigenvector information is sufficient for our purposes.

\paragraph{Self-reducibility in the space of eigenvalues.}

This suggests a natural self-reducible structure for the unitary orbit $\mathcal{O}_\Lambda$ via the principal submatrices.
The \emph{Rayleigh map} $\mathcal{R}$ (Definition \ref{def:rayleigh}) maps a matrix $X \in \mathcal{O}_\Lambda$ to the length-$\binom{n+1}{2}$ vector of the eigenvalues of all the leading principal minors of $X$.
These eigenvalues are organized in the form of a triangle called the \emph{Rayleigh triangle}, denoted
$\mathcal{R}(X) = (R_{i,j})_{1 \leq i \leq j \leq n}$ where $R_{i,j}$ is the $i^\text{th}$ largest eigenvalue of the top left $j \times j$ principal submatrix.
(Note that $R_{\bullet,n}$ are the eigenvalues of $\Lambda$, and counter to matrix indexing intuition, we refer to $R_{\bullet,n}$ is the \emph{top row} of the Rayleigh triangle.)
Organizing eigenvalues into a triangle like this then makes the self-reducible structure clear: fixing the top $n-k+1$ rows of the triangle, $R_{\bullet,n},\ldots,R_{\bullet,k}$, and leaving the bottom $k$ rows free gives a lower-dimensional space of Rayleigh triangles that corresponds precisely to the $\U(k)$-orbit of the $k \times k$ matrix $\mathrm{diag}(R_{\bullet,k})$.

We now have a self-reducible way to view the elements $X \in \mathcal{O}_\Lambda$ in terms of their eigenvalues, but how does this help us to sample from $\mathcal{O}_\Lambda$?
By the Cauchy--Rayleigh interlacing theorem, for all $X \in \mathcal{O}_\Lambda$ the Rayleigh triangle $\mathcal{R}(X)$ is a element of a polytope in $\R^{\binom{n+1}{2}}$ cut out by the inequalities
\vspace{-0.3em}
\[
    \vspace{-0.3em}
    R_{i,j} \geq R_{i,j-1} \geq R_{i+1,j} \quad \text{for all valid $i,j$}.
\]
In fact the converse is also true: the image $\mathcal{R}(\mathcal{O}_\Lambda)$ is the whole polytope cut out by these inequalities, called the \emph{Gelfand--Tsetlin (GT) polytope} and denoted $GT(\lambda)$ where $\lambda$ is the vector of eigenvalues of $\Lambda$.
What is special about the Rayleigh map $\mathcal{R}$ is then that it projects the uniform measure $d\mu_\Lambda$ on $\mathcal{O}_\Lambda$ to the Lebesgue (uniform) measure on the GT polytope.
This leaves a few questions.
\begin{enumerate}
    \item How does the Rayleigh map $\mathcal{R}$ project the HCIZ density from the unitary orbit to $GT(\lambda)$?
    \item How do we sample from the corresponding distribution on the GT polytope?
    \item How do we transfer that sample back to the unitary orbit?
\end{enumerate}
The answers here are reasonable: for $(1)$ the Rayleigh map projects the exponential HCIZ density to an exponential density on the GT polytope (see Theorem \ref{thm:equivalence}), for $(2)$ we can use powerful tools (\cite{LovaszVempala06} and \cite{BassilySmithThakurta14}) to sample from this exponential density on a polytope cut out by polynomially many inequalities, and for $(3)$ we have algorithms that utilize the symmetry of the orbit $\mathcal{O}_\Lambda$ and the HCIZ density (see Section \ref{sec:alg-descrption}).
For the understanding of the reader, we first demonstrate this explicitly in the case of rank-one projections.

\subsection{The case of rank-one projections}

Let us now look at the simplest choice of $\Lambda$: the case where $\Lambda$ is the diagonal matrix with entries $1,0,0,\ldots,0$.
This means that $\mathcal{O}_\Lambda$ is the set of Hermitian positive semidefinite (PSD) rank-one projections.
In this case, each leading principal submatrix of a given $X \in \mathcal{O}_\Lambda$ has at most one non-zero eigenvalue.
Thus, the entries of the Rayleigh triangle of $X$ are all zero except for $R_{1,j}$, for which we have $$1 = R_{1,n} \geq R_{1,n-1} \geq \cdots \geq R_{1,1} \geq 0.$$
This means that the GT polytope in this case is isomorphic to a simplex by considering the values of $R_{1,n}-R_{1,n-1}$, $\ldots$, $R_{1,2}-R_{1,1}$, and $R_{1,1}$, which sum to 1.
Since the principal submatrices are all rank at most 1, these differences are the differences of the traces of the submatrices of $X$, which are precisely equal to the diagonal entries of $X$.
Hence the map $\diag(X)$, which picks out the diagonal entries of $X$, is equivalent to the Rayleigh map $\mathcal{R}$ in this case, and the image $\diag(\mathcal{O}_\Lambda)$ is the standard simplex $\Delta_n$.
Therefore $\diag(X)$ maps the uniform measure $\mu_\Lambda(X)$ to the Lebesgue (uniform) measure on $\Delta_n$, and we can determine the measure on the GT polytope corresponding to the HCIZ density via
\[
    e^{\langle Y, X \rangle} d\mu_\Lambda(X) = e^{\langle y, \diag(X) \rangle} d\mu_\Lambda(X) \xrightarrow{\diag} e^{\langle y, x \rangle} dx,
\]
where $y$ is the vector of diagonal entries of $Y$ (which is itself a diagonal matrix).
This suggests a method for sampling from our
distribution $e^{\langle Y, X \rangle} d\mu_\Lambda(X)$ on $\mathcal{O}_\Lambda$ when $Y$ is diagonal:
\begin{enumerate}
    \item Sample $x$ from $\Delta_n$ according to the distribution $e^{\langle y, x \rangle}\, dx$ where $y = \diag(Y)$.
    \item Convert $x$ into a rank-one PSD projection $X \in \mathcal{O}_\Lambda$.
\end{enumerate}
%
This sampling problem is a special case of sampling from a log-concave density on a convex polytope, and there is a significant body of work which is geared towards coming up with algorithms for this general problem.
In particular, to obtain a algorithm which gives the TV distance bound promised by Theorem \ref{thm:main}, we can appeal to Corollary 1.2 of \cite{LovaszVempala06}.
To obtain an algorithm which gives the infinity divergence bound promised by Theorem \ref{thm:main2}, we appeal to Lemma 6.5 of the arXiv version of \cite{BassilySmithThakurta14}.
Since our function is log-linear, to use these results we just need to establish bounds on the outer and inner radius of the polytope (the simplex in this case) and the Lipschitz constant of the exponent.
This is trivial in the case of the simplex, but it is also straightforward for the case of a general GT polytope; see Section \ref{sec:algo_step1}.

However, we still need to convert $x$ into a rank-one PSD projection $X \in \mathcal{O}_\Lambda$ in a way which is compatible with our exponentially weighted distribution on $\mathcal{O}_\Lambda$.
Towards this, we first observe that the set of all $X$ for which $\diag(X) = x$ is given by
\vspace{-0.3em}
\[
    \vspace{-0.3em}
    \diag^{-1}(x) = \{X \in \mathcal{O}_\Lambda ~:~ X = vv^* \text{ where } v^* = (e^{-i\theta_1}\sqrt{x_1}, \ldots, e^{-i\theta_n}\sqrt{x_n})\}.
\]
Second, for diagonal $Y$ the density $e^{\langle Y, X \rangle}$ does not depend on the choice of $\theta_1,\ldots,\theta_n$.
Therefore if we restrict our distribution on $\mathcal{O}_\Lambda$ to the subset $\diag^{-1}(x)$, we obtain the uniform distribution.
This means we can convert $x$ into a rank-one PSD projection by uniformly randomly sampling $e^{i\theta_1},\ldots,e^{i\theta_n}$ independently from the unit circle and setting $X := vv^*$ where $$v^* = (e^{-i\theta_1}\sqrt{x_1}, \ldots, e^{-i\theta_n}\sqrt{x_n}).$$
Combining these two steps---sampling from the simplex and then transferring that sample to $\mathcal{O}_\Lambda$---gives us an algorithm for sampling from $\mathcal{O}_\Lambda$ according to our exponentially weighted density.

What remains is then to show that our bounds between the sampled and target distributions on the simplex (either TV distance or infinity divergence) transfer back to the respective distributions on $\mathcal{O}_\Lambda$.
In this case, the algorithm to transfer a sample from the simplex to a sample from $\mathcal{O}_\Lambda$ is very simple and explicit, as demonstrated above.
In particular, given a point $x$ of the simplex, we can sample \emph{exactly} from the target (uniform) distribution on the associated fiber of $x$ in $\mathcal{O}_\Lambda$.
This means that transferring from the simplex to $\mathcal{O}_\Lambda$ accumulates no extra error between the sampled and target distributions.
Therefore the TV distance and infinity divergence bounds between the sampled and target distributions on $\mathcal{O}_\Lambda$ are \emph{exactly} equal to the bounds achieved on the simplex $\Delta_n$ (see Appendix \ref{app:disintegration} for more details).
Thus we have achieved our desired error bounds for the sampled distribution on $\mathcal{O}_\Lambda$, completing the proof of the main results in the rank-one case.

\paragraph{Obstacles to extending to general $\Lambda$.}
Unfortunately, extending this algorithm beyond the rank-one case immediately runs into issues.
First, we need to know how the HCIZ density on a general orbit $\mathcal{O}_\Lambda$ transfers to the GT polytope through the Rayleigh map.
In the case of rank-one projections, we obtained a log-linear density which was crucial to our sampling error bounds, and we need to be able to emulate this in the general case.

Beyond this, converting a sample from the GT polytope back to the unitary orbit is now more complicated.
In the case of rank-one projections, determining the fiber $\diag^{-1}(x) \cong \mathcal{R}^{-1}(x)$ was straightforward and led to a simple method for sampling an element of the orbit $\mathcal{O}_\Lambda$.
For general $\Lambda$, the fiber $\mathcal{R}^{-1}(x)$ does not have such a clear description.
We will need to study further the relationship between the unitary orbit $\mathcal{O}_\Lambda$ and the corresponding GT polytope to understand how to generalize the sampling technique used for rank-one projections.

\subsection{Moving to the GT polytope in the case of general $\Lambda$}

In the case of general $\Lambda$, we can utilize the same overarching algorithm that was used in the rank-one case: sample from the GT polytope, and then transfer back to the unitary orbit.
The first question we need to answer is what the distribution on the GT polytope should look like.
We know that the Rayleigh map transfers the uniform distribution $d\mu_\Lambda(X)$ on the unitary orbit to the uniform distribution on the polytope, but what about the distribution $e^{\langle Y, X \rangle} d\mu_\Lambda(X)$?

Since $Y$ is diagonal we can write
$$
    e^{\langle Y, X \rangle} d\mu_\Lambda(X) = e^{\langle y, \diag(X) \rangle} d\mu_\Lambda(X),
$$
where $y = \diag(Y)$.
In the rank-one case, the Rayleigh map $\mathcal{R}$ was equivalent to the $\diag(X)$ map, and this meant that the projected measure was given by $e^{\langle y, x \rangle} dx$.
To handle the general case, we need a map which takes a Rayleigh triangle $R \in GT(\lambda)$ to the diagonal vector of the corresponding $X \in \mathcal{O}_\Lambda$.
This is precisely the $\type(R)$ map (Definition \ref{def:type}), which computes differences of the traces of the successive principal submatrices:
$$
    \type(R) = \Big( R_{1,1},\, R_{1,2} + R_{2,2} - R_{1,1},\, \hdots,\, \sum_{i=1}^n R_{i,n} - \sum_{j=1}^{n-1} R_{j,n-1} \Big).
$$
This definition implies $\type(R) = \diag(X)$ whenever $R = \mathcal{R}(X)$.
With this, we can more precisely state our sampling algorithm at a high level.
\begin{enumerate}
    \item Sample a Rayleigh triangle $R = (R_{i,j})$ from the associated GT polytope according to the distribution $e^{\langle y, \type(R) \rangle}\, dR$, where $y = \diag(Y)$.
    \item Convert $R$ into an element $X \in \mathcal{O}_\Lambda$ of the unitary orbit.
\end{enumerate}
As in the simplex case, we can use the powerful tools of \cite{LovaszVempala06} and \cite{BassilySmithThakurta14} to sample from the log-linear density on the GT polytope (see Section \ref{sec:algo_step1} for more details).
However, we still must convert this
into a sample from the HCIZ density on the orbit $\mathcal{O}_\Lambda$.

\subsection{From the GT polytope back to the unitary orbit}

Supposing we have a sample $R$ distributed according to the log-linear density on the GT polytope, the final step is to convert $R$ into a sample from the unitary orbit $\mathcal{O}_\Lambda$.
In the case of the simplex, this was easy because it is easy to describe the fiber $\diag^{-1}(x)$ as well as the uniform distribution on this fiber.
In the case of the GT polytope and the Rayleigh map however, determining $\mathcal{R}^{-1}(R)$ and the associated distribution is more complicated.

Fortunately though, we can break the problem down into more manageable pieces corresponding to the row-by-row self-reducible structure discussed above.
Observe that for any $k$ we can identify $\U(k-1)$ with the subgroup of $\U(k)$ consisting of unitary matrices that have a 1 in the bottom right corner and zeros in all other positions of the last row and column.
Inducting on this observation, we obtain inclusions $$\U(1) \hookrightarrow \U(2) \hookrightarrow \cdots \hookrightarrow \U(n-1) \hookrightarrow \U(n),$$ and these inclusions correspond precisely to sub-triangles of our sampled Rayleigh triangle $R$ (the bottom $1, 2, \ldots, n-1, n$ rows of the triangle respectively).

This allows us to induct on the rank of the unitary group.
Assuming that we have an $(n-1) \times (n-1)$ matrix sample $X_0$ from the $\U(n-1)$ orbit associated to the bottom $n-1$ rows of $R$, we just need to sample an $X \in \mathcal{O}_\Lambda$ which has $X_0$ as its top-left principal submatrix.
That is, given such an $X_0$, we need to sample some
$$
    X =
    \left[\begin{matrix}
        X_0 & v \\
        v^* & c
    \end{matrix}\right]
    \in \mathcal{O}_\Lambda.
$$
To sample such an $X$, first note that the top row of $R$ (that is, the eigenvalues of $\Lambda$) determines the trace of $X$, which specifies deterministically the value of $c$.
This leads to a crucial observation: all possible values of the matrix $X$ have the same diagonal entries, and hence our density function $e^{\langle Y, X \rangle}$ is constant since $Y$ is a diagonal matrix.
This means that we may sample $X$ \emph{uniformly} from the set of all $X$ with the above block form.

Having made this observation, we now describe in detail how to sample such a matrix $X$ (see Section \ref{sec:algo_step2}).

\paragraph{The case of distinct eigenvalues.}
We first demonstrate how to do this in a simplified case: when $X_0$ is a diagonal matrix with distinct eigenvalues.
In this case, we make the following easy observation for $U \in \U(n)$:
\[
    U
    \left[\begin{matrix}
        X_0 & v \\
        v^* & c
    \end{matrix}\right]
    U^* =
    \left[\begin{matrix}
        X_0 & w \\
        w^* & c
    \end{matrix}\right]
    \text{ for some $v,w$} ~ \iff ~ U \text{ is diagonal} ~ \iff ~ U = \diag(e^{i\theta_1}, \ldots, e^{i\theta_n}).
\]
This immediately gives rise to an algorithm for sampling $X$ of the above block form:
\begin{enumerate}
    \item Construct \emph{any} $X \in \mathcal{O}_\Lambda$.
    \item Sample $e^{i\theta_1}, \ldots, e^{i\theta_n}$ uniformly and independently from the unit circle.
    \item Defining $U := \diag(e^{i\theta_1}, \ldots, e^{i\theta_n})$, our sample is then $UXU^*$.
\end{enumerate}
What remains to be done then is to construct some $X \in \mathcal{O}_\Lambda$, which is equivalent to contructing a valid value of $v$.
To this end, we use the special form of $X$ to write down its characteristic polynomial.
Letting $X_0 = \diag(R_{1,n-1}, \ldots, R_{n-1,n-1})$, we want to choose $v$ such that
\[
    \prod_{i=1}^n (t - \lambda_i) = 
    \det\left(tI - 
    \left[\begin{smallmatrix}
        X_0 & v \\
        v^* & c
    \end{smallmatrix}\right]
    \right) = (t-c) \prod_{i=1}^{n-1} (t-R_{i,n-1}) + \sum_{i=1}^{n-1} |v_i|^2 \prod_{j \neq i} (t-R_{j,n-1}).
\]
Since the values of $R_{\bullet,n-1}$ are distinct, we obtain $n-1$ equations by plugging in $t = R_{k,n-1}$ for each value of $k \in \{1,\ldots,n-1\}$:
\[
    \prod_{i=1}^n (R_{k,n-1} - \lambda_i) = |v_k|^2 \prod_{i \neq k} (R_{k,n-1} - R_{i,n-1}) \implies |v_k|^2 = \frac{\prod_{i=1}^n (R_{k,n-1} - \lambda_i)}{\prod_{i \neq k} (R_{k,n-1} - R_{i,n-1})}.
\]
This gives us a formula for a choice of $v_k$, so long as the right-hand side is non-negative.
In fact, it is always non-negative because the values of $\lambda_\bullet$ and $R_{\bullet, n-1}$ are interlacing by the Cauchy--Rayleigh theorem.
By choosing $v_k \geq 0$ which satisfy the above equalities, we have constructed a valid $X$ from our orbit, and applying the above algorithm gives the desired sample of $\mathcal{O}_\Lambda$.

\paragraph{The general case.}
Handling the cases of non-distinct eigenvalues and non-diagonal $X_0$ is then straightforward.
First, if $X_0$ is diagonal with non-distinct ordered eigenvalues, the set of unitary matrices which preserves the block form of $X$ becomes larger:
\[
\begin{split}
    \text{distinct case:} \quad &UXU^* \in \mathcal{O}_\Lambda ~ \iff ~ U \text{ is diagonal} ~ \iff ~ U \in \U(1) \times \U(1) \times \cdots \times \U(1) \\
    \text{non-distinct case:} \quad &UXU^* \in \mathcal{O}_\Lambda ~ \iff ~ U \in \U(m_1) \times \U(m_2) \times \cdots \times \U(m_p) \times \U(1),
\end{split}
\]
where $m_1,m_2,\ldots,m_p$ are the multiplicities of the eigenvalues of $X_0$.
That is, we simply need to replace step 2 of the above algorithm with
\begin{enumerate}
    \setcounter{enumi}{1}
    \item Sample $U_1,\ldots,U_p,U_{p+1}$ uniformly from $\U(m_1),\ldots,\U(m_p),\U(1)$ respectively.
\end{enumerate}
Algorithms to sample uniformly from unitary groups are well-known and were discussed above.

Finally, handling the non-diagonal case is even easier.
Letting $U_0 \in \U(n-1)$ be such that $U_0 X_0 U_0^* = D_0$ is diagonal, we reduce to the previous cases by considering
\[
    \left[\begin{matrix}
        U_0 & 0 \\
        0 & 1
    \end{matrix}\right]
    X
    \left[\begin{matrix}
        U_0^* & 0 \\
        0 & 1
    \end{matrix}\right]
    =
    \left[\begin{matrix}
        U_0 X_0 U_0^* & U_0 v \\
        (U_0 v)^* & c
    \end{matrix}\right]
    =
    \left[\begin{matrix}
        D_0 & U_0 v \\
        (U_0 v)^* & c
    \end{matrix}\right].
\]
We then first sample a matrix $X'$ by applying the above algorithm to the right-hand side matrix above, since $D_0$ is diagonal.
We then obtain our desired sample via the inverse conjugation by $U_0$:
\[
    X = \left[\begin{matrix}
        U_0^* & 0 \\
        0 & 1
    \end{matrix}\right]
    X'
    \left[\begin{matrix}
        U_0 & 0 \\
        0 & 1
    \end{matrix}\right].
\]
Combining all of this then yields an algorithm which constructs a matrix $X$ in the unitary orbit $\mathcal{O}_\Lambda$ from the given Rayleigh triangle $R$ in the GT polytope.

\paragraph{Sampling error bounds.}
The last thing we must do is show that our bounds between the sampled and target distributions on the GT polytope (either TV distance or infinity divergence) transfer back to the respective distributions on $\mathcal{O}_\Lambda$.
In the case of the simplex, the exact bounds transfered from the simplex to the orbit because we were able to exactly sample from the fibers $\mathcal{R}^{-1}(x)$ for any $x$ in the simplex.
Specifically, sampling from the fiber boiled down to sampling uniformly from a torus.

In the general case, we saw above that this torus sampling in the case of the simplex is replaced by an inductive sampling of unitary matrices from the uniform (Haar) distribution.
There is a simple exact algorithm for sampling Haar-distributed unitary matrices, as discussed in Section \ref{sec:uniform_case}.
Therefore the same argument applies to the general case as applied to the case of rank-one projections and the simplex.
Specifically, the bounds we achieve between the sampled and target distributions on the GT polytope transfer back exactly to the respective distributions on the orbit $\mathcal{O}_\Lambda$.

This concludes the overview of the proof of Theorem \ref{thm:main}. To summarize, we first considered the pushforward measure of the HCIZ distribution on a unitary orbit $\mathcal{O}_\Lambda$ through the Rayleigh map $\mathcal{R}$ onto the GT polytope $GT(\lambda)$.
This gave rise to a density function on the GT polytope that was a log-linear function of the type vector.
We then sampled a Rayleigh triangle $R \in GT(\lambda)$ according to this density using very general techniques for sampling from log-linear distributions on convex polytopes.
Finally, we converted this sample $R \in GT(\lambda)$ into a sample from the orbit $\mathcal{O}_\Lambda$ by inductively sampling $k \times k$ matrices according to the bottom $k$ rows of $R$.
We refer the reader to Section \ref{sec:sampling_algo} for the remaining details of the proof.

\subsection{Differentially private low-rank approximation}
The proof of Theorem \ref{thm:diff-main} relies on the sampling algorithm in Theorem \ref{thm:main2}. 
In particular, to obtain the privacy guarantee, we use the exponential mechanism due to McSherry and Talwar \cite{MTalwar}.
We first note that, given a $d \times d$ Hermitian matrix $A$ with eigenvalues $\gamma_1 \geq \cdots \geq \gamma_d$, a $k \leq d$, and an $\eps$, the exponential mechanism requires us to sample according to the density
$ e^{\eps \langle A,P\rangle}$ on the set of rank-$k$ PSD projections.
The fact that the resulting $P$ is $\eps$-differentially private follows from the fact that sensitivity of the exponent $\langle A,P\rangle $ with respect to a change in $A$ is bounded by $1$ (Lemma \ref{lem:sensitivity}); see Lemma \ref{lem:DP}. 

To obtain a bound on the utility of the exponential mechanism proposed above, we use a covering argument (see Lemma \ref{lem:utility}). 
The main observation is that the space of $d \times d$ Hermitian rank-$k$ projection matrices can be covered by at most $(1 + \frac{8}{\zeta})^{2dk}$ balls of radius $\zeta$.
This $dk$ bound is better than the naive $d^2$ bound and allows us to prove the utility is at least $(1-\delta) \sum_{i=1}^k \gamma_i$ as long as $$\sum_{i=1}^k \gamma_i \geq C \cdot \frac{dk}{\epsilon \delta} \cdot \log \frac{1}{\delta}$$ for an absolute constant $C > 0$.
 
Finally, the number of arithmetic operations required follows directly from Theorem \ref{thm:main2} with the correct chosen parameters.
Specifically, we choose $\Lambda = \diag(1,\ldots,1,0,\ldots,0)$ with $k$ 1's and $n-k$ 0's, which implies $\mathcal{O}_\Lambda$ is the set of rank-$k$ PSD projections.
Hence, $$\lambda_1 - \lambda_d = 1-0 = 1$$ in this case.
This immediately implies the number of arithmetic operations required by the algorithm is polynomial in $d$, $\gamma_1 - \gamma_d$, $\frac{1}{\epsilon}$, and the number of bits required to represent $\gamma$ (as stated in the theorem).

\section{The sampling algorithms of Theorems \ref{thm:main} and \ref{thm:main2}} \label{sec:sampling_algo}

In this section, we describe the main steps of the algorithms claimed in Theorems \ref{thm:main} and \ref{thm:main2}.
We then prove that the steps produce the correct output and determine the number of arithmetic operations they require.
We will mostly treat the two algorithms together, since they differ only in the procedure used to sample from the Gelfand--Tsetlin polytope.

\begin{remark}
    Throughout we assume that we can exactly unitarily diagonalize Hermitian matrices for convenience.
    That said, the algorithms given in \cite{pan1999complexity} approximate the eigenvalues and eigenvectors of a Hermitian $n \times n$ matrix $H$ within relative error $2^{-L_H}$ in a number of arithmetic operations which is polynomial in $n$ and $\log L_H$, where $L_H$ is the number of bits required to represent $H$.
    To use this result, there is some extra error accounting which is required.
    In our case this can be handled, and we omit the details.
    
    Further, under this assumption we may also assume that the matrix $Y$ which appears in the exponent of our density function $e^{\langle Y, X \rangle}$ is not diagonal but just Hermitian with eigenvalues $y_1, \ldots, y_n$.
    Indeed if $Y = U \cdot \diag(y) \cdot U^*$, then the fact that $\mu_\Lambda(X)$ is unitarily invariant implies $e^{\langle Y, X \rangle} d\mu_\Lambda(X) = e^{\langle \diag(y), U^*XU \rangle} d\mu_\Lambda(X) = e^{\langle \diag(y), X \rangle} d\mu_\Lambda(X)$.
    After sampling $X$ according to this distribution on $\mathcal{O}_\Lambda$, we then simply conjugate $X$ by $U$ to obtain a sample from the original target distribution on $\mathcal{O}_\Lambda$.
\end{remark}

\subsection{Description of the algorithms} \label{sec:alg-descrption}

Formally, the input and output of the algorithms are as follows.

\begin{itemize}
    \item \textbf{Input:}
    \begin{enumerate}
        \item A vector $\lambda = (\lambda_1, \hdots, \lambda_n) \in \R^n$, with $\lambda_1 \geq \cdots \geq \lambda_n$.
        \item A vector $y=(y_1, \ldots, y_n) \in \mathbb{R}^n$, with $y_1 \geq \cdots \geq y_n$.
    \end{enumerate}
    \noindent
    Write $\Lambda = \diag(\lambda)$, $Y = \diag(y)$.
    \item \textbf{Output:} An $n \times n$ Hermitian matrix with eigenvalues $(\lambda_1, \hdots, \lambda_n)$, distributed according to $d\nu(X) \propto e^{\langle Y, X \rangle } d \mu_\Lambda(X)$ on $\mathcal{O}_\Lambda$.
\end{itemize}

\noindent
At a high level, both algorithms then consist of the following steps.

\begin{enumerate}
    \item {\bf Reduce sampling from $\mathcal{O}_\Lambda$ to sampling from $GT(\lambda)$.} Construct a membership oracle for $GT(\lambda)$ and an evaluation oracle for the correct exponential density function on $GT(\lambda)$.
    \item {\bf Sample a Rayleigh triangle from $GT(\lambda)$.} Sample a Rayleigh triangle $P = (P_{j,k})_{1 \le k \le j \le n}$ from the density proportional to $e^{\langle y, \mathrm{type}(P)\rangle}$ on the polytope $GT(\lambda)$.
    \item {\bf Sample from the fiber over $P$.} Sample a matrix $S$ uniformly at random from the fiber $\mathcal{R}^{-1}(P) = \{X \in \mathcal{O}_\Lambda \ | \ \mathcal{R}(X) = P \}$.
    \item {\bf Output} $S$.
\end{enumerate}

\noindent
We now describe Steps 1 and 3 in detail, and we also discuss the algorithms we cite and invoke for Step 2.
In Section \ref{sec:thm_proofs}, we then complete the proofs of Theorems \ref{thm:main} and \ref{thm:main2}.
Before describing the steps, we give one result which demonstrates that the steps of the above algorithm sample correctly from the orbit under the assumption that, in Step 2, we are able to sample exactly from the desired distribution with no error.
In Section \ref{sec:thm_proofs}, we will refer to the results of Appendix \ref{app:disintegration} for details on handling the case where the distribution on $GT(\lambda)$ only approximates the target distribution.

\begin{proposition}[Correctness of the ideal algorithm] \label{prop:correct_ideal}
Let $P$ be a random Rayleigh triangle distributed according to the distribution given by the density proportional to $e^{\langle y, \mathrm{type}(P)\rangle}$ on $GT(\lambda)$, and let $S$ be a uniform random element of $\mathcal{R}^{-1}(P)$.  
Then $S$ is distributed according to the measure $\nu(X) \propto e^{\langle Y, X \rangle } d \mu_\Lambda(X)$ on $\mathcal{O}_\Lambda$.
\end{proposition}

\begin{proof}
    By Theorem \ref{thm:equivalence}, the density function $e^{\langle Y, X \rangle }$ is constant on the fibers of $\mathcal{R}$, since $e^{\langle Y, X \rangle } = e^{\langle y, \mathrm{type}(\mathcal{R}(X))\rangle}$. 
    Using Theorem \ref{thm:equivalence} again, the statement then follows immediately from the disintegration theorem for probability measures; see Appendix \ref{app:disintegration} and \cite{chang1997}. 
\end{proof}

\subsection{Step 1: Reduce sampling from the orbit to sampling from the GT polytope} \label{sec:sampling_reduction}

In this section we describe the algorithm for constructing membership and evaluation oracles for the (unnormalized) exponential density function on the polytope $GT(\lambda)$.

Recall the following system of $n$ equalities and $n \times (n-1)$ inequalities which determine if a Rayleigh triangle $P$ is an element of $GT(\lambda)$ (see Equations \ref{eq:toprow} and \ref{eqn:interlacing-relns}):
\[
    \lambda_i= P_{i,n} \qquad \textrm{for all } 1 \le i \le n
\]
and
\[
    P_{i,j} - P_{i,j-1} \geq 0  \qquad \mbox{and} \qquad  P_{i,j-1}- P_{i+1,j} \geq 0  \qquad \textrm{for all } 1 \le i < j \le n.
\]
Note that whenever some of the values of $\lambda_i$ are actually equal, some of the inequalities will become equalities.
In particular, if $\lambda_p = \lambda_{p+1} = \cdots = \lambda_q$ then for any $P \in GT(\lambda)$ we have that $P_{i,j} = \lambda_p$ for all $i,j$ such that $p \leq i \leq q+j-n$.
Using this observation, for any fixed $\lambda$ it is straightforward to determine the ambient affine space in which $GT(\lambda)$ has non-empty interior.

The unnormalized density function on the polytope is then also easily described.
Given a real vector $y$, Theorem \ref{thm:equivalence} implies that the density function on $GT(\lambda)$ that we want to sample from is proportional to
\[
    f_0(P) = e^{\langle y, \type(P) \rangle}.
\]
Recall from Definition \ref{def:type} that $\type(P) \in \R^n$ is defined by
\[
    \type(P)_k := \sum_{i=1}^k P_{i,k} - \sum_{j=1}^{k-1} P_{j,k-1}.
\]
Using this definition, we can write down the exponent of $f_0$ as a linear functional on $P$.
We first have
\[
    \langle y, \type(P) \rangle = \sum_{k=1}^n y_n \left(\sum_{i=1}^k P_{i,k} - \sum_{j=1}^{k-1} P_{j,k-1}\right) = y_n \cdot \left[\sum_{i=1}^n \lambda_i\right] + \sum_{k=1}^{n-1} (y_k - y_{k+1}) \cdot \left[\sum_{i=1}^k P_{i,k}\right].
\]
Notice that for fixed $y$ and $\lambda$, we have that $y_n \cdot \left[\sum_{i=1}^n \lambda_i\right]$ is a constant in $P$.
Therefore we can push this part of the exponent into the normalization factor.
(For other entries of $P_{i,j}$ that are fixed by equalities in the $\lambda$ vector, this can also be done.)
We now define a triangle of values via $y^\Delta_{i,j} := y_j - y_{j+1}$ for $1 \leq i \leq j \leq n-1$.
With this, we want to sample from a density function on $GT(\lambda)$ proportional to
\begin{equation} \label{eq:log-linear}
    f(P) = e^{\langle y^\Delta, P \rangle} \qquad \text{with} \qquad y^\Delta_{i,j} := y_j - y_{j+1},
\end{equation}
where the top row of $P$ is ignored.
In particular, this means that the density function we want to sample from on $GT(\lambda)$ is in fact log-linear.
Note further that shifting $y$ by a multiple of the all-ones vector does not change the value of $y^\Delta$.
Therefore we may assume that $y_1 \geq 0 \geq y_n$ if desired. 

The above discussion then implies the following algorithmic guarantees for construction and running of the oracles.

\begin{lemma}[Membership and exact evaluation oracles] \label{lem:samp_oracles}
    There exists an algorithm such that, given $n \in \N$, $\lambda \in \R^n$, and $y \in \R^n$, outputs a membership oracle for $GT(\lambda)$ and an evaluation oracle for $f(P) = e^{\langle y, \type(P) \rangle}$.
    The number of arithmetic operations required to run this algorithm is polynomial in $n$ and the number of bits required to represent $y$ and $\lambda$.
    Further, the number of arithmetic operations required to then run these oracles with input $P$ is polynomial in the number of bits needed to represent $\lambda$, $y$, and $P$.
\end{lemma}

\subsection{Step 2: Sample a Rayleigh triangle from the GT polytope} \label{sec:algo_step1}

To sample a Rayleigh triangle $P \in GT(\lambda)$ according to the log-linear density discussed above, we appeal to powerful tools for sampling from log-concave and log-Lipschitz densities on convex polytopes.
Here we discuss two particular ways to do this, in terms of TV distance error and in terms of infinity divergence error.

\paragraph{Sampling from $GT(\lambda)$ with TV distance error.}

To sample from a distribution within TV distance $\xi$ from the target exponential density on $GT(\lambda)$, we appeal to a result of Lov\'asz and Vempala.\footnote{The precise statement that we invoke here is not explicitly stated in their papers, but follows readily from the cited results and has been confirmed to us in personal correspondence \cite{Vempala}.}

\begin{theorem}[Follows from Corollary 1.2 of \cite{LovaszVempala06}; see also Section 2.1 of \cite{LV06Simulated}] \label{thm:LV}
    Let $K \subset \R^d$ be a convex polytope, and for $\ell \in \R^d$ let $\mu_\ell$ denote the distribution on $K$ defined by the density function proportional to $f_\ell(x) := e^{\langle \ell, x \rangle}$.
    There is an algorithm that, given a membership oracle for $K$, a vector $\ell \in \R^d$, a point $x_0 \in K$, an outer radius $R$ of $K$, an inner radius $r$ of $K$, and a $\xi > 0$, samples from a distribution $\tilde{\mu}_\ell$ on $K$ with the property that
    \[
        \|\tilde{\mu}_\ell - \mu_\ell\|_\mathrm{TV} < \xi.
    \]
    The algorithm makes $\poly(d, \log \|\ell\|, \log \frac{R}{r}, \log \frac{1}{\xi})$ calls to the membership and evaluation oracles.
\end{theorem}

\noindent
Given the membership and evaluation oracles from Step 1 of the algorithm above, we can apply this result to sample from $GT(\lambda)$.
Beyond the oracles, we also need the starting point $x_0$ and outer and inner balls for $GT(\lambda)$, which we discuss below.

\paragraph{Sampling from $GT(\lambda)$ with infinity divergence error.}

To achieve the infinity divergence bound claimed in Theorem \ref{thm:main2}, we use a result of Bassily, Smith, and Thakurta.
We state a simplified version of this result here for the convenience of the reader.
Note that the dependence on $\log\frac{1}{r}$ appears because of the need to first put $GT(\lambda)$ in isotropic position; see Section 3.2 of the arXiv version of \cite{BassilySmithThakurta14}.

\begin{theorem}[\cite{BassilySmithThakurta14}, see Lemma 6.5] \label{thm:BST}
    Let $K \subset \R^d$ be a convex polytope, and for $\ell \in \R^d$ let $\mu_\ell$ denote the distribution on $K$ defined by the density function proportional to $f_\ell(x) := e^{\langle \ell, x \rangle}$.
    There is an algorithm that, given a membership oracle for $K$, a vector $\ell \in \R^d$, an outer radius $R$ of $K$, an inner radius $r$ of $K$, and a $\xi > 0$, samples from a distribution $\tilde{\mu}_\ell$ on $K$ with the property that
    \[
        D_\infty(\tilde{\mu}_\ell \| \mu_\ell) < \xi.
    \]
    The algorithm makes $\poly(d, \|\ell\|, R, \log\frac{1}{r}, \frac{1}{\xi})$ calls to the membership and evaluation oracles.
\end{theorem}

\noindent
Given the membership and evaluation oracles from Step 1 of the algorithm above, we can apply this result to sample from $GT(\lambda)$.
Beyond the oracles, we also need the outer and inner balls for $GT(\lambda)$, which we discuss below.

Before moving on, we note the main distinction between Theorems \ref{thm:BST} and \ref{thm:LV} above.
That is, Theorem \ref{thm:BST} achieves a stronger notion of approximation of the target exponential density at the cost of a larger number of oracle calls.
Specifically, the number of oracle calls in Theorem \ref{thm:BST} depends polynomially on $\|\ell\|$, $R$, and $\frac{1}{\xi}$, whereas in Theorem \ref{thm:LV} these dependencies are polylogarithmic.
We leave it as an open problem whether or not one can achieve infinity divergence error of $\xi$ in $\poly(d, \log \|\ell\|, \log \frac{R}{r}, \log \frac{1}{\xi})$ oracle calls.

\paragraph{Extra inputs required for the polytope sampling algorithms.}

As discussed above, we also need to be able to compute some extra data to apply the above polytope sampling algorithms to the target exponential density function on $GT(\lambda)$.
Specifically, we need a starting point $P_0$ for the algorithm, an outer radius $R$, and an inner radius $r$.
We give this data in the following three lemmas.
Note that by using a simpler argument than that of Lemma \ref{lem:samp_inner_ball}, one can achieve a worse bound on $r$ which is good enough for our purposes; see Remark \ref{rem:simple_r_bound} below.

\begin{lemma}[Starting point for sampling from $GT(\lambda)$] \label{lem:samp_starting_point}
    There is an algorithm that, given $\lambda \in \R^n$, samples uniformly from the polytope $GT(\lambda)$.
    The number of arithmetic operations required to run this algorithm is polynomial  in the number of bits required to represent $\lambda$.
\end{lemma}
\begin{proof}
    We can achieve this by sampling a random unitary matrix $U$, conjugating $\Lambda = \diag(\lambda)$ by $U$ to get $H$, and then constructing the Rayleigh triangle $\mathcal{R}(H)$ associated to $H$.
    This last step requires diagonalizing all of the leading principal submatrices of $H$, which can be done with a number of arithmetic operations polynomial in $n$ and in the number of bits required to represent $\lambda$.
    The fact that this process produces a uniformly random sample from $GT(\lambda)$ then follows from Proposition \ref{prop:R-pushfwd}.
\end{proof}

\begin{lemma}[Outer ball for $GT(\lambda)$] \label{lem:samp_outer_ball}
    The polytope $GT(\lambda)$ is contained by a ball of radius $R = \sqrt{n} \cdot (\lambda_1 - \lambda_n)$.
\end{lemma}
\begin{proof}
    The definition of $GT(\lambda)$ (see Definition \ref{def:GT} and Equations \ref{eq:toprow} and \ref{eqn:interlacing-relns}) implies $\lambda_1 \geq P_{i,j} \geq \lambda_n$ for every $P \in GT(\lambda)$.
    Thus $GT(\lambda)$ is contained in an $\ell^\infty$-ball of radius $\lambda_1 - \lambda_n$.
    Therefore $GT(\lambda)$ is contained in an $\ell^2$-ball of radius $\sqrt{n} \cdot (\lambda_1 - \lambda_n)$.
\end{proof}

\begin{lemma}[Inner ball for $GT(\lambda)$] \label{lem:samp_inner_ball}
    Let $q > 0$ be the minimal integer such that $\lambda_i = \frac{p_i}{q}$ for some integers $p_1,\ldots,p_n$.
    The polytope $GT(\lambda)$, considered as a subset of its affine span, contains a ball of radius $r = \frac{1}{8n^2q}$.
\end{lemma}
\begin{proof}
    We now construct a Rayleigh triangle $P \in GT(\lambda)$ which will be the center of a small ball contained in $GT(\lambda)$.
    Our assumption on the $\lambda_i$ implies the top ($n$th) row of $P$ is filled integer multiples of $\frac{1}{q}$.
    Therefore we can fill in the free entries of the $(n-1)$st row of $P$ with integer multiples of $\frac{1}{2q}$ without saturating any of the defining inequalities for the GT polytope.
    Now including the fixed entries of the $(n-1)$st row of $P$, this implies the $(n-1)$st row of $P$ is filled with integer multiples of $\frac{1}{2q}$.
    Therefore we can fill in the free entries of the $(n-2)$nd row of $P$ with integer multiples of $\frac{1}{3q}$ without saturating any of the defining inequalities for the GT polytope.
    Continuing this inductively, we can fill in the entries of the $(n-j)$th row of $P$ with integer mutliples of $\frac{1}{(j+1)q}$ without saturating any of the inequalities for the GT polytope, for every $j$.
    
    We now claim that an $\infty$-norm ball of radius $\frac{1}{8n^2q}$ centered at $P$ is contained in $GT(\lambda)$.
    To see this, we want to show that for any entry $P_{i,n-j}$ we have
    \[
        P_{i,n-j+1} + \frac{1}{8n^2q} \leq P_{i,n-j} - \frac{1}{8n^2q}, \quad P_{i-1,n-j-1} + \frac{1}{8n^2q} \leq P_{i,n-j} - \frac{1}{8n^2q},
    \]
    \[
        P_{i,n-j} + \frac{1}{8n^2q} \leq P_{i+1,n-j+1} - \frac{1}{8n^2q}, \quad P_{i,n-j} + \frac{1}{8n^2q} \leq P_{i,n-j-1} - \frac{1}{8n^2q},
    \]
    whenever the indices are valid.
    After removing the $\pm\frac{1}{8n^2q}$ terms, these are precisely the inequalities involving $P_{i,n-j}$ which define the GT polytope.
    Thus, we know that the inequalities are strict without the $\pm\frac{1}{8n^2q}$ terms.
    That is, for some integers $k_0,k_1,k_2,k_3,k_4$ we have
    \[
        \frac{k_1}{jq} = P_{i,n-j+1} < P_{i,n-j} = \frac{k_0}{(j+1)q}, \quad \frac{k_2}{(j+2)q} = P_{i-1,n-j-1} < P_{i,n-j} = \frac{k_0}{(j+1)q},
    \]
    \[
        \frac{k_0}{(j+1)q} = P_{i,n-j} < P_{i+1,n-j+1} = \frac{k_3}{jq}, \quad \frac{k_0}{(j+1)q} = P_{i,n-j} < P_{i,n-j-1} = \frac{k_4}{(j+2)q},
    \]
    whenever the indices are valid.
    For example, the first inequality implies
    \[
    \begin{split}
        \frac{k_1 \cdot (j+1)}{j(j+1)} < \frac{k_0 \cdot j}{j(j+1)} &\implies \frac{k_1 \cdot (j+1) + \frac{1}{2}}{j(j+1)} \leq \frac{k_0 \cdot j - \frac{1}{2}}{j(j+1)} \\
            &\implies \frac{k_1 \cdot (j+1)}{j(j+1)} + \frac{1}{2(2n)^2} \leq \frac{k_0 \cdot j}{j(j+1)} - \frac{1}{2(2n)^2} \\
            &\implies P_{i,n-j+1} + \frac{1}{8n^2q} \leq P_{i,n-j} - \frac{1}{8n^2q}
    \end{split}
    \]
    since $k_1 \cdot (j+1)$ and $k_0 \cdot j$ are integers.
    The same argument applies to all 4 inequalities, and this completes the proof.
\end{proof}

\begin{remark} \label{rem:simple_r_bound}
    One can obtain a cheaper bound on the radius $r$ of a small ball contained in $GT(\lambda)$, by defining $P \in GT(\lambda)$ inductively by simply choosing $P_{i,j}$ to be the midpoint between $P_{i,j+1}$ and $P_{i+1,j+1}$ for all valid $i,j$.
    Using this as the center of a small ball, one obtains a bound of $r \geq \frac{1}{2^{O(n)} q}$.
    Since the number of arithmetic operations required by our algorithms depends polylogarithmically on $\frac{1}{r}$, this bound would be enough for our purposes.
\end{remark}

\subsection{Step 3: Sample a uniform random matrix from the fiber} \label{sec:algo_step2}

Once we have sampled $P$, it remains to sample a matrix $S$ uniformly at random from the fiber $\mathcal{R}^{-1}(P)$.  The uniform distribution on the fiber is defined by the property that if $X$ is uniformly distributed in $\mathcal{R}^{-1}(P)$, then for $1 < k \le n$, $X[k]$ is uniformly distributed on the compact manifold $\mathcal{H}(X[k-1]; P_{\bullet,k})$ of $k \times k$ Hermitian matrices with eigenvalues $P_{\bullet,k}$ and leading $(k-1) \times (k-1)$ submatrix equal to $X[k-1]$.  Equivalently, the uniform measure on $\mathcal{R}^{-1}(P)$ is the disintegration (via the Rayleigh map) of the uniform measure on $\mathcal{O}_\Lambda$, in the sense of Theorem \ref{thm:disintegration}.
We construct a uniform sample $S \in \mathcal{R}^{-1}(P)$ using an inductive procedure, successively sampling the last row and column of each leading submatrix $S[k]$.  
We first define $S[1]$ to be the $1 \times 1$ matrix $[P_{1,1}]$, and we then sample each submatrix $S[k]$, for $1 < k \le n$, such that $S[k]$ is uniformly distributed on $\mathcal{H}(S[k-1]; P_{\bullet,k})$.  Explicitly, we sample $S[k]$ given $S[k-1]$ as follows.

\paragraph{Sampling procedure for $S[k]$ given $S[k-1]$ and $P_{\bullet,k}$:}

\begin{enumerate}
    \item \textbf{Diagonalize $S[k-1]$:} Compute a unitary matrix $U \in \U(k-1)$ such that $U^* \cdot S[k-1] \cdot U$ is diagonal.
    \item \textbf{Compute the new diagonal entry of $S[k]$:} Write
    \begin{equation} \label{eqn:Mk-def}
    S[k] = \begin{bmatrix} S[k-1] & Uv \\ (Uv)^* & c \end{bmatrix},
    \end{equation}
    where $v \in \C^{k-1}$ and $c \in \R$ are to be determined. Since the diagonal entries of $S[k]$ are just the type vector of $\mathcal{R}(S[k])$, we can compute
    $$c = \sum_{i = 1}^k P_{i,k} - \sum_{j = 1}^{k-1} P_{j,k-1}.$$
    
    \item \textbf{Compute the magnitudes of the new off-diagonal entries of $S[k]$:} It remains to sample $v$ uniformly at random from the set of vectors in $\C^{k-1}$ such that the matrix $S[k]$ in (\ref{eqn:Mk-def}) has spectrum $P_{\bullet,k}$. 
    We prove below in Lemma \ref{lem:samp_fiber_correctness} that this can be done using the following procedure. 
    Let $\delta_1 > \cdots > \delta_m$ be the distinct entries of $P_{\bullet,k-1}$, where $\delta_i$ has multiplicity $n_i$, so that $n_1 + \cdots + n_m = k-1$. 
    The interlacing relations (\ref{eqn:interlacing-relns}) imply that each value $\delta_i$ occurs in $P_{\bullet,k}$ with multiplicity at least $n_i - 1$.  
    Let $(\mu_1, \hdots, \mu_{m+1})$ be the vector obtained by removing $n_i - 1$ entries equal to $\delta_i$ from $P_{\bullet,k}$, for each $i$.  
    Then define
    $$r_i = \sqrt{- \frac{\prod_{j=1}^{m+1} (\delta_i - \mu_j)}{\prod_{j \not = i} (\delta_i - \delta_j)}}.$$
    The interlacing relations guarantee that the quantity under the square root above is nonnegative, so that $r_i$ is well defined.  As shown below, the vector $v$ is distributed uniformly on a product of complex spheres of radii $r_1, \hdots, r_m$.
    
    \item \textbf{Sample the phases of the new off-diagonal entries of $S[k]$:} For each $i = 1, \hdots, m$, we then sample the $n_i$ coordinates
    $$\left( v_{1 + \sum_{j = 1}^{i-1} n_j}, \ \hdots, \ v_{n_i + \sum_{j = 1}^{i-1} n_j} \right) \in \C^{n_i}$$
    uniformly at random from the sphere of radius $r_i$ in $\C^{n_i}$. This last step can be accomplished by well-known methods; see e.g.~\cite{Muller-spherepoints}.
\end{enumerate}

\paragraph{Output:}

Finally, after iteratively sampling all of the leading submatrices, we output $S = S[n]$.

\paragraph{Correctness and number of operations of the iterative algorithm.}

We now prove that the above algorithm samples from the correct distribution on the fiber $\mathcal{R}^{-1}(P)$ of $P$, and then we bound the number of operations the algorithm requires.

\begin{lemma}[Sampling from the fiber over $P$: Correctness] \label{lem:samp_fiber_correctness}
    The above algorithm, given a Rayleigh triangle $P \in GT(\lambda)$, returns a uniform random element of the fiber $\mathcal{R}^{-1}(P) = \{S \in \mathcal{O}_\Lambda : \mathcal{R}(S) = P\}$.
\end{lemma}
\begin{proof}
    Again we write $\mathcal{H}(S[k-1]; P_{\bullet,k})$ for the set of $k \times k$ Hermitian matrices with eigenvalues $P_{\bullet,k}$ and $(k-1)$th leading submatrix equal to $S[k-1]$.  It only remains to show that $\mathcal{H}(S[k-1]; P_{\bullet,k})$ is a product of spheres as described above.  Specifically, let $U \in \U(k-1)$ be a unitary matrix diagonalizing $S[k-1]$, so that $U^* \cdot S[k-1] \cdot U = \diag(P_{\bullet,k-1}).$  We will show
    \begin{equation} \label{eqn:H-sphere-prod}
    \mathcal{H}(S[k-1]; P_{\bullet,k}) = \Bigg \{ \begin{bmatrix} S[k-1] & Uv \\ (Uv)^* & c \end{bmatrix} \ \bigg | \ v \in \C^{k-1}, \ \sum_{l = 1}^{n_i} |v_{l + \sum_{j = 1}^{i-1} n_j}|^2 = r_i^2 \textrm{ for } i=1, \hdots, m \Bigg \},
    \end{equation}
    where we necessarily have $$c = \sum_{i = 1}^k P_{i,k} - \sum_{j = 1}^{k-1} P_{j,k-1}$$ due the the fact that the diagonal of any Hermitian matrix $X$ is equal to $\mathrm{type}(\mathcal{R}(X))$. Write $D = \diag(P_{\bullet, k-1})$. To establish (\ref{eqn:H-sphere-prod}), we must show that a matrix of the form
    \[
        S = 
        \begin{bmatrix}
            U & 0 \\ 0 & 1
        \end{bmatrix}
        \begin{bmatrix}
            D & v \\ v^* & c
        \end{bmatrix}
        \begin{bmatrix}
            U^* & 0 \\ 0 & 1
        \end{bmatrix}
        =
        \begin{bmatrix}
            S[k-1] & Uv \\ (Uv)^* & c
        \end{bmatrix}
    \]
    has eigenvalues $P_{\bullet,k}$ if and only if $\sum_{l = 1}^{n_i} |v_{l + \sum_{j = 1}^{i-1} n_j}|^2 = r_i^2$ for $i = 1, \hdots, m$.  We prove this by writing the characteristic polynomial of $S$ in two different ways.  First, if $S$ has eigenvalues $P_{\bullet,k}$ then
    \begin{equation} \label{eqn:charpoly1}
    \det(tI - S) = \prod_{i=1}^k (t - P_{i,k}).
    \end{equation}
    On the other hand, we must have
    \[
        \det(tI - S) = 
        \det\left(tI - \begin{bmatrix}
            D & v \\ v^* & c
        \end{bmatrix}\right) = \det\begin{bmatrix}
            tI-D & v \\ v^* & t-c
        \end{bmatrix},
    \]
    and expanding along the first row and column we find that this equals
    \begin{equation} \label{eqn:charpoly2}
    (t-c) \prod_{j=1}^m (t-\delta_j)^{n_j} - \sum_{i=1}^m \left(|v_{n_1 + \cdots + n_{i-1}+1}|^2 + \cdots + |v_{n_1 + \cdots + n_i}|^2\right) (t-\delta_i)^{n_i-1} \prod_{j \neq i} (t-\delta_j)^{n_j}.
    \end{equation}
    We have $S \in \mathcal{H}(S[k-1]; P_{\bullet,k})$ exactly when (\ref{eqn:charpoly1}) equals (\ref{eqn:charpoly2}).  Equating these two expressions for the characteristic polynomial and recalling that interlacing of $P_{\bullet,k}$ and $P_{\bullet,k-1}$ implies that $P_{\bullet,k}$ contains the value $\delta_i$ with multiplicity at least $n_i - 1$ for all $i$, we can divide through both sides by $(t-\delta_i)^{n_i-1}$ for all $i$ to obtain
    \begin{equation} \label{eqn:charpoly-both}
        \prod_{i=1}^{m+1} (t-\mu_i) = (t-c) \prod_{j=1}^m (t-\delta_j) - \sum_{i=1}^m \left(|v_{n_1 + \cdots + n_{i-1}+1}|^2 + \cdots + |v_{n_1 + \cdots + n_i}|^2\right) \prod_{j \neq i} (t-\delta_j).
    \end{equation}
    Note that both sides of (\ref{eqn:charpoly-both}) are monic polynomials of degree $m+1$, and
    \[
        c = \sum_{i=1}^k P_{i,k} - \sum_{j=1}^{k-1} P_{j,k-1} = \sum_{i=1}^{m+1} \mu_i - \sum_{j=1}^m \delta_j
    \]
    implies the coefficients of $t^{m+1}$ and $t^m$ on both sides are equal. Therefore the polynomials in (\ref{eqn:charpoly-both}) are equal if and only if they are equal at $m$ distinct points.  Evaluating both sides at $t = \delta_i$ for $i = 1, \hdots, m$, we find that they are equal exactly when
    $$\sum_{l = 1}^{n_i} |v_{l + \sum_{j = 1}^{i-1} n_j}|^2 = r_i^2$$
    for $i = 1, \hdots, m$, which is the desired result.
\end{proof}

\begin{lemma}[Sampling from the fiber over $P$: Number of operations] \label{lem:samp_fiber_runningtime}
    The number of arithmetic operations the above algorithm requires to sample uniformly from the fiber $\mathcal{R}^{-1}(P)$ is polynomial  in the number of bits required to represent the entries of $P$.
\end{lemma}
\begin{proof}
    To determine the number of arithmetic operations required by this part of the algorithm, we first determine the number of operations for each of the steps described above.
    Step 1 amounts to unitarily diagonalizing a Hermitian matrix $H = U^*DU$, and this can be done in a number of operations which is polynomial in the size of the matrix $k$ and the bit complexity of the entries of the matrix $H$.
    (We will discuss below the bit complexity of $H$.)
    Step 2 involves basic matrix operations with $U$ and elements of the Rayleigh triangle $P$, which again depends polynomially on $k$ and the bit complexity of $H$ and $P$.
    Step 3 does basic arithmetic on the entries of $P$, requiring a number of operations which is polynomial in $k$ and the bit complexity of the entries of the Rayleigh triangle $P$.
    Step 4 requires sampling of elements of the unit sphere and multiplying those samples by the magnitudes computed in step 3, and this also can be done in a number of operations which is polynomial in $k$ and the bit complexity of the entries of $P$.
    
    The whole iterative process to construct $S = S[n]$ then requires $n$ iterations of the above 4 steps, where the output to each iteration is $S[k]$ and the input to each iteration is the Rayleigh triangle $P$ along with the output $S[k-1]$ of the previous iteration.
    Note that steps 2, 3, and 4 only refer to the entries of $P$ and not to the output of the previous iteration.
    The new entries of $S[k]$ constructed from steps 2, 3, and 4 then require $\poly(n, L_P)$ bits to represent, where $L_P$ is the number of bits needed to represent the entries of $P$.
    Thus in each iteration we add new entries, with bit complexity $\poly(n, L_P)$, to $S[k-1]$ to construct the output $S[k]$.
    The unitary diagonalization of $S[k-1]$ in step 1 then requires a number of operations which is polynomial in the number of bits needed to represent the entries of $S[k-1]$.
    And since we are only adding new entries to $S[k-1]$ to construct $S[k]$ (not changing previously constructed entries), after each iteration the entries of $S[k]$ require $\poly(n, L_P)$ bits to represent.
    After all $n$ iterations, the algorithm has sampled $S = S[n]$ in a number of arithmetic operations that is polynomial in $n$ and in the number of bits required to represent the entries of $P$.
\end{proof}

\section{Proofs of Theorems \ref{thm:main} and \ref{thm:main2}} \label{sec:thm_proofs}

In this section, we complete the proofs of Theorems \ref{thm:main} and \ref{thm:main2} using the results of the previous section.
We first prove correctness of the algorithms, and then we prove bounds on the required number of arithmetic operations.

\subsection{Correctness} \label{sec:correctness_overall}

For Theorem \ref{thm:main}, we want to show that the algorithm from the previous section samples from $\mathcal{O}_\Lambda$ according to the exponential density proportional to $e^{\langle Y, X \rangle} d\mu_\Lambda(X)$.
Recall that the algorithm consists of two main steps: sampling $P$ from $GT(\lambda)$ (called Step 2 above), and then sampling from the fiber $\mathcal{R}^{-1}(P)$ (called Step 3 above).
To sample from $GT(\lambda)$, we use one of two algorithms: the algorithm given by Theorem \ref{thm:LV} for TV distance error claimed in Theorem \ref{thm:main}, or the algorithm given by Theorem \ref{thm:BST} for infinity divergence error claimed in Theorem \ref{thm:main2}.
These algorithms require a membership oracle for $GT(\lambda)$ (given by Lemma \ref{lem:samp_oracles}), a vector $\ell$ and an evaluation oracle for the target density $g(P) \propto e^{\langle \ell, P \rangle}$ on $GT(\lambda)$ (also given by Lemma \ref{lem:samp_oracles}), a starting point for the algorithm (given by Lemma \ref{lem:samp_starting_point}), and outer and inner balls for $GT(\lambda)$ (given by Lemmas \ref{lem:samp_outer_ball} and \ref{lem:samp_inner_ball} respectively).

Once we have our sample $P$ from $GT(\lambda)$, we use it to sample uniformly from the fiber $\mathcal{R}^{-1}(P)$ via Lemma \ref{lem:samp_fiber_correctness}.
The last thing we need to prove then is that, by sampling from $GT(\lambda)$ and then from the corresponding fiber, we are in fact sampling from $\mathcal{O}_\Lambda$ according to the exponential density proportional to $e^{\langle Y, X \rangle} d\mu_\Lambda(X)$ as claimed.
For this, we handle the cases of Theorems \ref{thm:main} and \ref{thm:main2} separately.

\paragraph{Correctness for Theorem \ref{thm:main}.}
Let $\nu$ be the target distribution on $GT(\lambda)$ associated to the unnormalized density function $f(P) = e^{\langle y, \type(P) \rangle}$.
Equation \ref{eq:log-linear} shows that $f(P) = e^{\langle y^\Delta, P \rangle}$, and thus we can apply Theorem \ref{thm:LV} to $f$ to sample $P$ from $GT(\lambda)$ according to a distribution $\hat{\nu}$ for which $\|\hat{\nu}-\nu\|_{\mathrm{TV}} < \xi$.

Now given $P$ in $GT(\lambda)$, Lemma \ref{lem:samp_fiber_correctness} then says that the algorithm of Section \ref{sec:algo_step2} samples uniformly from the fiber of $P$.
By Corollary \ref{cor:mulambda-disint}, the uniform distribution on the fiber $\mathcal{R}^{-1}(P)$ is the disintegrated measure of the target distribution on $\mathcal{O}_\Lambda$ (see Appendix \ref{app:disintegration} for more discussion).
Lemma \ref{lem:TV-lift} then implies the overall algorithm samples from $\mathcal{O}_\Lambda$ according to a distribution which is within TV distance error $\xi$ of the target.
(See also Proposition \ref{prop:correct_ideal} for a similar result in the ideal case.)
This completes the proof of correctness of the algorithm of Theorem \ref{thm:main}.

\paragraph{Correctness for Theorem \ref{thm:main2}.}
Let $\nu$ be the target distribution on $GT(\lambda)$ associated to the unnormalized density function $f(P) = e^{\langle y, \type(P) \rangle}$.
Equation \ref{eq:log-linear} shows that $f(P) = e^{\langle y^\Delta, P \rangle}$, and thus we can apply Theorem \ref{thm:BST} to $f$ to sample $P$ from $GT(\lambda)$ according to a distribution $\hat{\nu}$ for which $D_\infty(\hat{\nu}\|\nu) < \xi$.

Given $P$ in $GT(\lambda)$, Lemma \ref{lem:samp_fiber_correctness} then says that the algorithm of Section \ref{sec:algo_step2} samples uniformly from the fiber of $P$.
As above, Corollary \ref{cor:mulambda-disint} says the uniform distribution on the fiber $\mathcal{R}^{-1}(P)$ is the disintegrated measure of the target distribution on $\mathcal{O}_\Lambda$ (see Appendix \ref{app:disintegration} for more discussion).
Lemma \ref{lem:alpha-lift} then implies our algorithm samples from $\mathcal{O}_\Lambda$ according to a distribution which is within infinity divergence error $\xi$ of the target.
(See also Proposition \ref{prop:correct_ideal} for a similar result in the ideal case.)
This completes the proof of correctness of the algorithm of Theorem \ref{thm:main2}.

\subsection{Number of operations}

We now determine the number of arithmetic operations required of the algorithms of Theorems \ref{thm:main} and \ref{thm:main2}.
For both algorithms we are given $n \in \N$, $\lambda \in \R^n$, $y \in \R^n$, and a desired error bound $\xi > 0$.
As described in Section \ref{sec:alg-descrption}, we need to (1) construct the membership oracles, (2) use them to sample $P$ from the polytope $GT(\lambda)$, and (3) then sample uniformly from the fiber $\mathcal{R}^{-1}(P)$ over $P$.
Steps 1 and 3 are exactly the same for both algorithms.
Lemma \ref{lem:samp_oracles} implies the necessary oracles can be constructed using a number of operations which is polynomial in $n$ and in the number of bits required to represent $y$ and $\lambda$.
Lemma \ref{lem:samp_fiber_runningtime} implies we can sample from the fiber over $P$ in a number of operations which is polynomial in $n$ and in the number of bits required to represent $P$.
We now discuss the number of arithmetic operations required of the algorithms used to sample $P$ from $GT(\lambda)$.

Recall from Section \ref{sec:sampling_reduction} that the target distribution on $GT(\lambda)$ is given by a density proportional to $f(P) = e^{\langle y^\Delta, P \rangle}$ with $y^\Delta_{i,j} := y_j - y_{j+1}$.
From this we achieve the bound
\[
    \|y^\Delta\| \leq n^2 (y_1-y_n).
\]
Further, we also have the outer and inner balls for $GT(\lambda)$ via Lemmas \ref{lem:samp_outer_ball} and \ref{lem:samp_inner_ball}, given as
\[
    R = \sqrt{n} \cdot (\lambda_1-\lambda_n) \qquad \text{and} \qquad \frac{1}{r} = 8n^2q,
\]
where $q > 0$ is an integer such that $\lambda_i = \frac{p_i}{q}$ for some integers $p_1,\ldots,p_n$.
We now use these bounds in order to finish the analysis of the algorithms.

\paragraph{Number of operations for Theorem \ref{thm:main}.}
For Theorem \ref{thm:main}, we apply Theorem \ref{thm:LV} as described above which implies we can sample from a distribution on $GT(\lambda)$ within TV distance error $\xi$ of the target distribution in $\poly(n, \log \frac{1}{\xi}, \log(y_1-y_n), \log(\lambda_1-\lambda_n), L_y, L_\lambda)$ calls to the membership and evaluation oracles, where $L_y$ and $L_\lambda$ are the number of bits required to represent $y$ and $\lambda$ respectively.
Since $\log(y_1-y_n)$ and $\log(\lambda_1-\lambda_n)$ are bounded above by $L_y$ and $L_\lambda$ respectively, we have that the above sampling can be done in $\poly(n, \log \frac{1}{\xi}, L_y, L_\lambda)$ calls to the membership and evaluation oracles.
The bits then required to represent the sample $P$ from $GT(\lambda)$ can then be no larger than the number of oracle calls.
Combining this with Lemma \ref{lem:samp_oracles} and the above discussion implies the number of arithmetic operations required to run the algorithm claimed by Theorem \ref{thm:main} is polynomial in $n$, $\log \frac{1}{\xi}$, and the number of bits needed to represent $y$ and $\lambda$.
Since the number of bits needed to represent $y$ (or $\lambda$) is at least $n$, we can drop the explicit dependence on $n$.

\paragraph{Number of operations for Theorem \ref{thm:main2}.}
For Theorem \ref{thm:main2}, we apply Theorem \ref{thm:BST} as described above which implies we can sample from a distribution on $GT(\lambda)$ within infinity divergence error $\xi$ of the target distribution in $\poly( y_1-y_n, \lambda_1-\lambda_n, \frac{1}{\xi}, L_y, L_\lambda)$ calls to the membership and evaluation oracles, where $L_y$ and $L_\lambda$ are the number of bits required to represent $y$ and $\lambda$ respectively.
As above, the bits then required to represent the sample $P$ from $GT(\lambda)$ can then be no larger than the number of oracle calls.
Combining this with Lemma \ref{lem:samp_oracles} and the above discussion implies the number of arithmetic operations required to run the algorithm claimed by Theorem \ref{thm:main2} is polynomial in $y_1-y_n$, $\lambda_1-\lambda_n$, $\frac{1}{\xi}$, and the number of bits needed to represent $y$ and $\lambda$.

\section{Differentially private rank-$k$ approximation} \label{app:diff-priv}

We consider the problem of differentially private low-rank approximation.
In the low-rank approximation problem, we are given a $d \times d$ real positive semidefinite (PSD) matrix $A$ and $1\leq k \leq n$, and the goal is to output the space spanned by the top $k$ eigenvectors of $A$.
Let $\mathcal{P}_k$ denote the set of $d \times d$ rank-$k$ Hermitian PSD projection matrices, considered as a subset of the space of complex Hermitian matrices.
It is easy to see  that 
$$ \max_{P \in \mathcal{P}_k } \langle P, A \rangle = \sum_{i=1}^{k} \gamma_i,$$
where $\gamma_1 \geq \cdots \geq \gamma_d \geq 0$ are the eigenvalues of $A$.

\paragraph{Differential privacy.} Let $\mathcal{U}$ be the universe of users. For each $u \in \mathcal{U}$, we have 
a vector $v_u \in \mathbb{R}^d$ such that $\|v_u\|_2 \leq 1$.
Given a dataset $D \subseteq \mathcal{U}$, define $A:=\sum_{u \in D} v_u v_u^*$.

\begin{definition}
Given an $\eps>0$ and a set $R$, a randomized mechanism $\mathcal{M}:\mathbb{R}^{d \times d} \to {R}$ is said to be $(\eps,0)$-differentially private if for all $S \subseteq R$ and for all $D,D' \subseteq \mathcal{U}$ such that the symmetric set difference $D \Delta D'$ has cardinality 2, one has 
$\Pr[\mathcal{M}(A) \in S] \leq e^{\eps} \Pr[\mathcal{M}(A') \in S].$
Here $A:=\sum_{u \in D} v_u v_u^*$ and  $A':=\sum_{u \in D'} v_u v_u^*$.
\end{definition}

\noindent
In our setting, $R$ is the space of $d \times d$ and rank-$k$ Hermitian matrices. 
We now copy Theorem \ref{thm:diff-main} from the introduction, which we prove in this section.

\begin{theorem}[Differentially private low-rank approximation]\label{thm:diff-main2}
    There is a randomized algorithm that, given a positive semidefinite $d \times d$ matrix $A$ and its eigenvalues $\gamma_1 \geq \cdots \geq \gamma_d$, an integer $1 \leq k \leq d$, and an $\eps>0$, outputs a rank-$k$ $d \times d$ Hermitian projection $P$ that is $(\eps,0)$-differentially private. 
    Moreover, there is a universal constant $C > 0$ such that, if there is a $\delta>0$ satisfying $\sum_{i=1}^k \gamma_i \geq C \cdot  \frac{dk}{\epsilon \delta} \cdot \log\frac{1}{\delta}$, then we have:
    \[
        \mathbb{E}_P\left[\langle A, P \rangle\right] \geq (1-\delta) \sum_{i=1}^k \gamma_i.
    \]
    The number of arithmetic operations required by this algorithm is polynomial in  $\frac{1}{\epsilon}$, $\gamma_1-\gamma_d$, and the number of bits needed to represent $\gamma$.
\end{theorem}

\noindent 
This result generalizes a Hermitian version of Theorem 1.1 of \cite{KTalwar}, where the above result is given in the case of $k=1$ for real symmetric rank-one matrices.
Specifically their Theorem 1.1 gives an algorithm which outputs an $(\epsilon,0)$-differentially private real unit vector $v$ for which the expected value of $v^\top A v = \langle A, vv^\top \rangle$ is bounded below by $(1-\delta)\gamma_1$ whenever $\gamma_1 \geq \Omega(\frac{d}{\epsilon\delta} \cdot \log \frac{1}{\delta})$.
They then use the rank-one case to prove a somewhat similar result in the general rank-$k$ case, which we state now.
The main difference here is that their Theorem 1.2 stated below outputs a real symmetric positive semidefinite matrix which approximates $A$, while our Theorem \ref{thm:diff-main2} above outputs a Hermitian projection $P$ which projects onto a $k$-dimensional subspace for which $\langle A, P \rangle$ approximates the sum of the top $k$ eigenvalues of $A$.

\begin{theorem}[Theorem 1.2 of \cite{KTalwar}]
    Let $A$ be a $d \times d$ real symmetric positive semidefinite matrix with eigenvalues $\gamma_1 \geq \cdots \geq \gamma_d$.
    There exists an $(\epsilon,0)$-differentially private polynomial-time algorithm for computing a matrix $A_k$ of rank at most $k$ so that $\|A-A_k\|_2 \leq \gamma_{k+1} + \delta \gamma_1$ as long as $\gamma_1 \geq \Omega(\frac{dk^3}{\epsilon \delta^6})$.
\end{theorem}

\noindent
We now compare the respective utility bounds for the two differentially private rank-$k$ mechanisms.
For their mechanism the ``utility'' can be described by the error term $\delta \gamma_1$, which is bounded below by
\[
   \delta \gamma_1 \geq \Omega\left(\frac{dk^3}{\epsilon \delta^5}\right).
\]
For our mechanism the ``utility'' can be described by the error term $\delta \sum_{i=1}^k \gamma_i$, which is bounded below by
\[
    \delta \sum_{i=1}^k \gamma_i \geq \Omega\left(\frac{dk}{\epsilon} \cdot \log \frac{1}{\delta}\right).
\]
Since $\delta$ is assumed to be small, our rank-$k$ mechanism improves upon the utility (error bound) of the rank-$k$ mechanism from \cite{KTalwar}.

\paragraph{The proof of Theorem \ref{thm:diff-main2}.}
We now prove Theorem \ref{thm:diff-main2} by combining the exponential mechanism framework due to \cite{MTalwar} with Theorem \ref{thm:main2}. 
Given a $\sigma > 0$, we define $\mathcal{M}' = \mathcal{M}'(A)$ to be the mechanism which is given by the sampling algorithm of Theorem \ref{thm:main2} with $\lambda$ being the vector that has $k$ ones and $d-k$ zeros, $y$ being the eigenvalues of $A$ multiplied by $\frac{\eps}{4\sigma}$, and $\xi=\frac{\eps}{2}$.
Therefore $\mathcal{M}' = \mathcal{M}'(A)$ outputs a sample from a distribution $\tilde{\nu}'_A$ on the set of $n \times n$ rank-$k$ PSD projections which is within infinity divergence error $\frac{\eps}{2}$ of the distribution $\nu'_A$ given by the density $e^{\frac{\eps}{4\sigma} \langle \diag(\gamma), P \rangle}$.

Next we diagonalize $A$ to determine the unitary matrix $U$ for which $A = U \cdot \diag(\gamma) \cdot U^*$.
With this, we define $\mathcal{M} = \mathcal{M}(A)$ to be the mechanism which is given by sampling $P$ from $\mathcal{M}'(A)$ and then outputting $UPU^*$.
Since $\langle A, UPU^* \rangle = \langle U \cdot \diag(\gamma) \cdot U^*, UPU^* \rangle = \langle \diag(\gamma), P \rangle$ and $\mathcal{P}_k$ is unitarily invariant, we have that $\mathcal{M}$ outputs a sample from a distribution $\tilde{\nu}_A$ on the set of $n \times n$ rank-$k$ PSD projections which is within infinity divergence error $\frac{\eps}{2}$ of the target distribution $\nu_A$ given by the density $e^{\frac{\eps}{4\sigma} \langle A, P \rangle}$.

Now suppose $\sigma$ is an upper bound on the following ``sensitivity'' of the function $\langle A, P \rangle$:
$$ \sup_{A,A'} \sup_{P} |\langle A, P \rangle - \langle A', P \rangle|,  $$ 
where $A,A'$ are such that $A'=A-v_1v_1^*+v_2v_2^*$ for some $\|v_1\|_2, \|v_2\|_2 \leq 1$ and $P$ is a rank-$k$ PSD projection matrix.
Then Lemma \ref{lem:sensitivity} says that we can choose $\sigma = 1$, and with this
Lemma \ref{lem:DP} implies that $\mathcal{M}$ is $(\eps,0)$-differentially private.

The number of arithmetic operations required for this algorithm then can be bounded by applying Theorem \ref{thm:main2} directly with our specified inputs.
The number of operation required in Theorem \ref{thm:main2} is polynomial in $d$, $\lambda_1 - \lambda_d$, $y_1 - y_d$, $\frac{1}{\epsilon}$, and the number of bits required to represent $\lambda$ and $y$.
In our case, $\lambda$ is a vector of 0's and 1's, and $y = \gamma$ is the sequence of eigenvalues of $A$.
Therefore the number of arithmetic operations required to run the algorithm is polynomial in $\frac{1}{\epsilon}$, $\gamma_1-\gamma_d$, and the number of bits needed to represent $\gamma$ (which is at least $d$).

\subsection{Correctness: Privacy guarantee}

The privacy guarantee given below in Lemma \ref{lem:DP} requires the following lemma on the sensitivity of the function $\langle A, P \rangle$.

\begin{lemma}[Sensitivity] \label{lem:sensitivity}
    For all $A,A'$ PSD  with $A'=A-v_1v_1^*+v_2v_2^*$ for some $v_1,v_2$ such that $\|v_1\|_2, \|v_2\|_2 \leq 1$ and for all rank-$k$ PSD projection matrices $P$, we have that
    \[
        |\langle A, P \rangle - \langle A', P \rangle| \leq 1 =: \sigma.
    \]
\end{lemma}

\begin{proof}
    We compute
    \[
        |\langle A, P \rangle - \langle A', P \rangle| = |\langle A-A', P \rangle| = |\langle v_1v_1^*-v_2v_2^*, P \rangle| = |v_1^* P v_1 - v_2^* P v_2| \leq 1.
    \]
    The inequality above follows from the fact that $v^* P v \in [0,1]$ for all vectors $v$ of norm at most 1.
\end{proof}

\begin{lemma}[Privacy via the exponential mechanism] \label{lem:DP}
For $\sigma=1$,
the mechanism $\mathcal{M}$ is $(\eps,0)$-differentially private.
\end{lemma}
\begin{proof}
Given positive definite $A$, define $A' = A - v_1v_1^* + v_2v_2^*$ for some $v_1,v_2$ such that $\|v_1\|_2, \|v_2\|_2 \leq 1$.
Given a rank-$k$ PSD projection $P$, we want to bound the ratio of the densities of $\mathcal{M}$ at $P$ with respect to $A$ and $A'$.
Let $\tilde{\nu}_A(P)$ denote the density of $P$ as outputted by the mechanism $\mathcal{M}(A)$, and let $\nu_A(P)$ denote the target density of the mechanism $\mathcal{M}(A)$, for which $D_\infty(\tilde{\nu}_A\|\nu_A) < \frac{\epsilon}{4}$ by definition of $\mathcal{M}$.
We now apply the sensitivity lemma stated above to the ideal densities to obtain
\[
\begin{split}
    \frac{\nu_A(P)}{\nu_{A'}(P)} = \frac{\frac{ e^{\frac{\epsilon}{4}  \langle A, P \rangle}}{\int_{Q \in \mathcal{P}_k} e^{\frac{\epsilon}{4}  \langle A, Q\rangle}d\mu_k(Q)}}{\frac{ e^{\frac{\epsilon}{4}  \langle A', P \rangle}}{\int_{Q \in \mathcal{P}_k} e^{\frac{\epsilon}{4}  \langle A', Q\rangle}d\mu_k(Q)}} &= e^{\frac{\epsilon}{4}  \langle A-A', P \rangle} \cdot \frac{\int_{Q \in \mathcal{P}_k} e^{\frac{\epsilon}{4}  \langle A', Q\rangle}d\mu_k(Q)}{\int_{Q \in \mathcal{P}_k} e^{\frac{\epsilon}{4}  \langle A, Q\rangle}d\mu_k(Q)} \\
        &\leq e^{\frac{\epsilon}{4}} \cdot \frac{\int_{Q \in \mathcal{P}_k} e^{\frac{\epsilon}{4}  \langle A, Q\rangle + \frac{\epsilon}{2} \langle A'-A, Q \rangle}d\mu_k(Q)}{\int_{Q \in \mathcal{P}_k} e^{\frac{\epsilon}{4}  \langle A, Q\rangle}d\mu_k(Q)} \\
        &\leq e^{\frac{\epsilon}{4}} \cdot \max_{Q \in \mathcal{P}_k} e^{\frac{\epsilon}{4} |\langle A'-A, Q \rangle|} \\
        &\leq e^{\frac{\epsilon}{2}}.
\end{split}
\]
Using the infinity divergence bounds between $\tilde{\nu}_A$ and $\nu_A$, we then further have that
\[
    \frac{\tilde{\nu}_A(P)}{\tilde{\nu}_{A'}(P)} = \frac{\tilde{\nu}_A(P) / \mu_A(P)}{\tilde{\nu}_{A'}(P) / \nu_{A'}(P)} \cdot \frac{\mu_A(P)}{\mu_{A'}(P)} \leq \frac{e^{\frac{\epsilon}{4}}}{e^{-\frac{\epsilon}{4}}} \cdot e^{\frac{\epsilon}{2}} = e^\epsilon.
\]
\end{proof}

\subsection{The utility bound}

The utility bound given below in Lemma \ref{lem:utility} requires the following lemma on the covering number for the orbit $\mathcal{P}_k$.

\begin{lemma}[Covering number for $\mathcal{P}_k$] \label{lem:covering_number}
    Let $\mathcal{P}_k$ denote the set of $d \times d$ rank-$k$ Hermitian PSD projection matrices, considered as a subset of the space of Hermitian matrices equipped with the $\ell^2$ operator norm.
    For any $\zeta > 0$, the number of balls centered in $\mathcal{P}_k$ of radius $\zeta$ required to cover the set $\mathcal{P}_k$ is at most $(1 + \frac{8}{\zeta})^{2dk}$.
\end{lemma}
\begin{proof}
    First consider $S_k$, the set of $k \times d$ complex matrices with orthonormal rows.
    Fix any $M \in S_k$ and let $U$ be a unitary matrix such that the first $k$ rows of $U$ are the rows of $M$.
    Letting $\|\cdot\|_2$ denote the $2 \to 2$ operator norm, we have that $\|M\|_2 = \|M^*\|_2 = 1$ since $M^*M$ is a PSD projection.
    Hence, the set $S_k$ can be considered a subset of the unit sphere in a $dk$-dimensional complex normed vector space.
    By a standard result, we can cover the complex unit ball in such a space with respect to any norm by at most $(1 + \frac{2}{\zeta})^{2dk}$ balls of radius $\zeta$ for any $\zeta > 0$.
    By replacing each such ball $B$ with a ball of radius $2\zeta$ centered about any $M \in B \cap S_k$ (if such a point exists), we have that we can cover $S_k$ by at most $(1 + \frac{4}{\zeta})^{2dk}$ balls centered in $S_k$ of radius $\zeta$ for any $\zeta > 0$.
    
    Now consider the map $\phi: M \mapsto M^* M$, which maps $S_k$ onto $\mathcal{P}_k$ the set of $d \times d$ rank-$k$ Hermitian PSD projections.
    Further, given $M,M' \in S_k$ such that $\|M-M'\|_2 < \frac{\zeta}{2}$, we have
    \[
        \|\phi(M)-\phi(M')\|_2 = \|M^*M - (M')^*M'\|_2 \leq \|M^*(M-M')\|_2 + \|(M-M')^*M'\|_2 \leq \frac{\zeta}{2} + \frac{\zeta}{2} = \zeta.
    \]
 Thus, for any $\frac{\zeta}{2}$-ball $B$ centered at some $M \in S_k$, we have that $\phi(B \cap S_k)$ is contained in an $\zeta$-ball centered at $M^*M$.
    Therefore since $\phi$ is surjective, $\mathcal{P}_k$ can be covered by at most $(1 + \frac{8}{\zeta})^{2dk}$ balls centered in $\mathcal{P}_k$ of radius $\zeta$ for any $\zeta > 0$.
\end{proof}

\begin{lemma}[Utility bound] \label{lem:utility}
    The rank-$k$ exponential mechanism, given a $d \times d$ Hermitian positive definite matrix $A$ with eigenvalues $\gamma_1 \geq \gamma_2 \geq \cdots \geq \gamma_d$, outputs a $d \times d$ rank-$k$ Hermitian PSD projection $P$ such that
    \[
        \mathbb{E}_P\left[\langle A, P \rangle\right] \geq (1-\delta) \sum_{i=1}^k \gamma_i
    \]
    as long as $\sum_{i=1}^k \gamma_i \geq C \cdot \frac{dk}{\epsilon \delta} \cdot \log\frac{1}{\delta}$ for small $\delta > 0$ and an absolute constant $C > 0$.
\end{lemma}
\begin{proof}
    We first define ``good'' and ``bad'' sets via
    \[
        G := \left\{P \in \mathcal{P}_k ~:~ \langle A, P \rangle \geq \left(1-\frac{\delta}{2}\right) \sum_{i=1}^k \gamma_i\right\}, \quad B := \left\{P \in \mathcal{P}_k ~:~ \langle A, P \rangle \leq (1-\delta) \sum_{i=1}^k \gamma_i\right\}.
    \]
    Let $P_0$ be the projection associated to the top $k$ eigenvectors of $A$, and define $A_0 := AP_0$ so that the top $k$ eigenpairs of $A_0$ agree with that of $A$ and the rest of the eigenvalues are 0.
    Now fix any $P \in \mathcal{P}_k$ such that $\|P-P_0\|_2 < \frac{\delta}{2}$, where $\|\cdot\|_2$ denotes the $\ell^2$ operator norm.
    Since the $\ell^2$ operator norm is the $\infty$-norm on the singular values, we can apply H\"older's inequality to get
    \[
        \langle A, P \rangle \geq \langle A_0, P \rangle = \langle A_0, P_0 \rangle - \langle A_0, P_0 - P \rangle \geq \sum_{i=1}^k \gamma_i - \|P_0 - P\|_2 \sum_{i=1}^k \gamma_i > \left(1-\frac{\delta}{2}\right) \sum_{i=1}^k \gamma_i.
    \]
    That is, every $P \in \mathcal{P}_k$ contained in the ball of radius $\frac{\delta}{2}$ about $P_0$ is also contained in $G$.
    
    Letting $\mu_k$ be the unitarily invariant probability measure on $\mathcal{P}_k$, the covering number lemma (Lemma \ref{lem:covering_number}) implies there is some ball $B_{\delta/2}(P')$ centered at $P' \in \mathcal{P}_k$ of radius $\frac{\delta}{2}$ is such that $\mu_k(B_{\delta/2}(P')) \geq (1 + \frac{16}{\delta})^{-2dk}$.
    By unitary invariance of $\mu_k$, we then have
    \[
        \mu_k(G) \geq \mu_k(B_{\delta/2}(P_0)) = \mu_k(B_{\delta/2}(P')) \geq e^{-2dk \log(1 + \frac{16}{\delta})} \geq e^{-C' \cdot dk \log\frac{1}{\delta}}
    \]
    for some absolute $C' \geq 2$ whenever $\delta$ is small.
    Now let
    \[
        f_A(P) := e^{\frac{\epsilon}{2} \langle A, P \rangle}
    \]
    denote the unnormalized probability density function of the exponential mechanism, and let $Z$ be the normalization constant.
    Then whenever $\sum_{i=1}^k \gamma_i \geq 8C' \cdot \frac{dk}{\epsilon \delta} \cdot \log\frac{1}{\delta}$, we have
    \[
    \begin{split}
        \frac{\mathbb{P}[P \in B]}{\mathbb{P}[P \in G]} &\leq \frac{\mu_k(B) \cdot \max_{P \in B} \frac{f_A(P)}{Z}}{\mu_k(G) \cdot \min_{P \in G} \frac{f_A(P)}{Z}} \\
            &\leq \frac{1 \cdot e^{\frac{\epsilon}{2} (1-\delta) \sum_{i=1}^k \gamma_i}}{e^{-C' \cdot dk \log\frac{1}{\delta}} \cdot e^{\frac{\epsilon}{2} (1-\frac{\delta}{2}) \sum_{i=1}^k \gamma_i}} = \frac{e^{\frac{\epsilon}{2}(-\frac{\delta}{2}) \sum_{i=1}^k \gamma_i}}{e^{-C' \cdot dk \log\frac{1}{\delta}}} \\
            &\leq \frac{e^{-2C' \cdot dk \log\frac{1}{\delta}}}{e^{-C' \cdot dk \log\frac{1}{\delta}}} = e^{-C' \cdot dk \log\frac{1}{\delta}}.
    \end{split}
    \]
    Therefore for $C' \geq 2$ and $\delta > 0$ small, we have
    \[
        \mathbb{P}[P \not\in B] \geq 1 - e^{-C' \cdot dk \log \frac{1}{\delta}} \cdot \mathbb{P}[P \in G] \geq 1 - \delta^2 \cdot \mathbb{P}[P \in G]
    \]
    which implies
    \[
    \begin{split}
        \mathbb{E}_P[\langle A, P \rangle] &\geq \mathbb{P}[P \in G] \cdot \left(1 - \frac{\delta}{2}\right) \sum_{i=1}^k \gamma_i + \mathbb{P}[P \not\in B,G] \cdot (1 - \delta) \sum_{i=1}^k \gamma_i + \mathbb{P}[P \in B] \cdot 0 \\
            &\geq \mathbb{P}[P \in G] \cdot \left(1 - \frac{\delta}{2}\right) \sum_{i=1}^k \gamma_i + \left[1 - \mathbb{P}[P \in G] \cdot \left(1 + \delta^2\right)\right] \cdot (1 - \delta) \sum_{i=1}^k \gamma_i \\
            &= (1-\delta) \sum_{i=1}^k \gamma_i + \mathbb{P}[P \in G] \cdot \left(\sum_{i=1}^k \gamma_i\right) \cdot \left[\left(1 - \frac{\delta}{2}\right) - (1-\delta) \cdot \left(1 + \delta^2\right)\right] \\
            &= (1-\delta) \sum_{i=1}^k \gamma_i + \mathbb{P}[P \in G] \cdot \left(\sum_{i=1}^k \gamma_i\right) \cdot \left[\frac{\delta}{2} - \delta^2 + \delta^3\right] \\
            &\geq (1-\delta) \sum_{i=1}^k \gamma_i
    \end{split}
    \]
    whenever $\delta > 0$ is small enough.
\end{proof}

\section*{Acknowledgements}
This research was supported in part by NSF CCF-1908347, NSF DMS-1714187, and JST CREST program JPMJCR18T6.
This research was also funded in part by the Deutsche Forschungsgemeinschaft (DFG, German Research  Foundation) under Germany's Excellence Strategy - The Berlin Mathematics  Research Center MATH+ (EXC-2046/1, project ID: 390685689). 
We would like to thank Ainesh Bakshi, Anay Mehrotra, Kunal Talwar, Abhradeep Thakurta, Enayat Ullah, and Oren Mangoubi for useful  discussions.
\bibliographystyle{plain}
\bibliography{refs}

\begin{thebibliography}{10}

\bibitem{Band1976}
William Band and James~L. Park.
\newblock New information-theoretic foundations for quantum statistics.
\newblock {\em Foundations of Physics}, 6(3):249--262, Jun 1976.

\bibitem{Barysh}
Yu. Baryshnikov.
\newblock {GUE}s and queues.
\newblock {\em Probab. Theory Relat. Fields}, 119:256--274, 2001.

\bibitem{BassilySmithThakurta14}
Raef Bassily, Adam Smith, and Abhradeep Thakurta.
\newblock Private empirical risk minimization: Efficient algorithms and tight
  error bounds.
\newblock In {\em 2014 IEEE 55th Annual Symposium on Foundations of Computer
  Science}, pages 464--473. IEEE, 2014.

\bibitem{BouKaz}
D.~V. Boulatov and V.~A. Kazakov.
\newblock The {I}sing model on a random planar lattice: the structure of the
  phase transition and the exact critical exponents.
\newblock {\em Physics Letters B}, 186:379--384, 1987.

\bibitem{chang1997}
Joseph~T Chang and David Pollard.
\newblock Conditioning as disintegration.
\newblock {\em Statistica Neerlandica}, 51(3):287--317, 1997.

\bibitem{Kamalika}
Kamalika Chaudhuri, Anand Sarwate, and Kaushik Sinha.
\newblock Near-optimal differentially private principal components.
\newblock In F.~Pereira, C.~J.~C. Burges, L.~Bottou, and K.~Q. Weinberger,
  editors, {\em Advances in Neural Information Processing Systems 25}, pages
  989--997. Curran Associates, Inc., 2012.

\bibitem{ChikuseBook}
Y.~Chikuse.
\newblock {\em Statistics on Special Manifolds}.
\newblock Lecture Notes in Statistics. Springer New York, 2012.

\bibitem{ChikusePaper}
Yasuko Chikuse.
\newblock Concentrated matrix {L}angevin distributions.
\newblock {\em Journal of Multivariate Analysis}, 85(2):375 -- 394, 2003.

\bibitem{DGZ}
P.~Di~Francesco, P.~Ginsparg, and J.~Zinn-Justin.
\newblock 2{D} gravity and random matrices.
\newblock {\em Physics Reports}, 254:1--133, 1995.
\newblock \url{http://arxiv.org/abs/hep-th/9306153}.

\bibitem{Eynard-note}
B.~Eynard.
\newblock A short note about {M}orozov's formula, 2004.
\newblock Service de Physique Th\'eorique de Saclay, report no. SPHT-T04-077.
  \url{https://arxiv.org/abs/math-ph/0406063}.

\bibitem{E-PF}
B.~Eynard and A.~Prats~Ferrer.
\newblock 2-matrix versus complex matrix model, integrals over the unitary
  group as triangular integrals.
\newblock {\em Commun. Math. Phys.}, 264:115--144, 2006.
\newblock \url{https://arxiv.org/abs/hep-th/0502041}.

\bibitem{AGsurvey}
A.~Guionnet.
\newblock Large deviations and stochastic calculus for large random matrices.
\newblock {\em Probability Surveys}, 1:72--172, 2004.
\newblock \url{https://arxiv.org/abs/math/0409277}.

\bibitem{HarishChandra1957}
Harish-Chandra.
\newblock Differential operators on a semisimple {L}ie algebra.
\newblock {\em American Journal of Mathematics}, 79(1):87--120, 1957.

\bibitem{IZ}
C.~Itzykson and J.-B. Zuber.
\newblock The planar approximation. {II}.
\newblock {\em Journal of Mathematical Physics}, 21:411--421, 1980.

\bibitem{Jaynes1}
Edwin~T. {Jaynes}.
\newblock {Information theory and statistical mechanics}.
\newblock {\em Physical Review}, 106:620--630, May 1957.

\bibitem{Jaynes2}
Edwin~T. {Jaynes}.
\newblock {Information theory and statistical mechanics. II}.
\newblock {\em Physical Review}, 108:171--190, October 1957.

\bibitem{JerrumVV86}
Mark Jerrum, Leslie~G. Valiant, and Vijay~V. Vazirani.
\newblock Random generation of combinatorial structures from a uniform
  distribution.
\newblock {\em Theor. Comput. Sci.}, 43:169--188, 1986.

\bibitem{KTalwar}
Michael Kapralov and Kunal Talwar.
\newblock On differentially private low rank approximation.
\newblock In {\em Proceedings of the Twenty-Fourth Annual ACM-SIAM Symposium on
  Discrete Algorithms}, SODA '13, page 1395–1414, USA, 2013. Society for
  Industrial and Applied Mathematics.

\bibitem{Kaz}
V.~A. Kazakov.
\newblock Ising model on a dynamical planar random lattice: exact solution.
\newblock {\em Physics Letters A}, 119:140--144, 1986.

\bibitem{LeakeV20b}
Jonathan Leake and Nisheeth~K. Vishnoi.
\newblock On the computability of continuous maximum entropy distributions:
  Adjoint orbits of {L}ie groups.
\newblock In {\em arXiv 2011.01851}, 2020.

\bibitem{LeakeV20}
Jonathan Leake and Nisheeth~K. Vishnoi.
\newblock On the computability of continuous maximum entropy distributions with
  applications.
\newblock In {\em Proceedings of the 52nd Annual ACM SIGACT Symposium on Theory
  of Computing}, STOC 2020, page 930–943, New York, NY, USA, 2020.
  Association for Computing Machinery.

\bibitem{LovaszVempala06}
L{\'a}szl{\'o} Lov{\'a}sz and Santosh Vempala.
\newblock Fast algorithms for logconcave functions: Sampling, rounding,
  integration and optimization.
\newblock In {\em 2006 47th Annual IEEE Symposium on Foundations of Computer
  Science (FOCS'06)}, pages 57--68. IEEE, 2006.

\bibitem{LV06Simulated}
L{\'a}szl{\'o} Lov{\'a}sz and Santosh Vempala.
\newblock Simulated annealing in convex bodies and an o*(n4) volume algorithm.
\newblock {\em Journal of Computer and System Sciences}, 72(2):392--417, 2006.

\bibitem{MTalwar}
F.~{McSherry} and K.~{Talwar}.
\newblock Mechanism design via differential privacy.
\newblock In {\em 48th Annual IEEE Symposium on Foundations of Computer Science
  (FOCS'07)}, pages 94--103, 2007.

\bibitem{McS-expository}
C.~McSwiggen.
\newblock The {H}arish-{C}handra integral: An introduction with examples, 2018.
\newblock \url{https://arxiv.org/abs/1806.11155}.

\bibitem{Morozov}
A.~Morozov.
\newblock Pair correlator in the {I}tzykson--{Z}uber integral.
\newblock {\em Modern Physics Letters A}, 7:3503--3507, 1992.
\newblock \url{https://arxiv.org/abs/hep-th/9209074}.

\bibitem{Muller-spherepoints}
M.~E. Muller.
\newblock A note on a method for generating points uniformly on
  {$N$}-dimensional spheres.
\newblock {\em Comm. Assoc. Comput. Mach.}, 2:19--20, 1959.

\bibitem{Neretin}
Yu.~A. Neretin.
\newblock Rayleigh triangles and non-matrix interpolation of matrix beta
  integrals.
\newblock {\em Sbornik: Mathematics}, 194(4):515--540, apr 2003.

\bibitem{NO}
N.~O'Connell.
\newblock Whittaker functions and related stochastic processes.
\newblock {\em MSRI Publications: Random Matrix Theory, Interacting Particle
  Systems and Integrable Systems}, 65:385--409, 2014.

\bibitem{pan1999complexity}
Victor~Y Pan and Zhao~Q Chen.
\newblock The complexity of the matrix eigenproblem.
\newblock In {\em Proceedings of the thirty-first annual ACM symposium on
  Theory of computing}, pages 507--516, 1999.

\bibitem{PEDZ}
A.~Prats~Ferrer, B.~Eynard, P.~Di~Francesco, and J.-B. Zuber.
\newblock Correlation functions of {H}arish-{C}handra integrals over the
  orthogonal and the symplectic groups.
\newblock {\em Journal of Statistical Physics}, 129:885--935, 2007.
\newblock \url{https://arxiv.org/abs/math-ph/0610049}.

\bibitem{Sha}
S.~L. Shatashvili.
\newblock Correlation functions in the {I}tzykson--{Z}uber model.
\newblock {\em Communications in Mathematical Physics}, 154:421--432, 1993.
\newblock \url{https://arxiv.org/abs/hep-th/9209083}.

\bibitem{SlaterEntropy1}
Paul~B. Slater.
\newblock Relations between the barycentric and von {N}eumann entropies of a
  density matrix.
\newblock {\em Physics Letters A}, 159(8):411 -- 414, 1991.

\bibitem{Tao}
Terrence Tao.
\newblock {\em The Harish-Chandra-Itzykson-Zuber integral formula}, 2013.
\newblock
  \url{https://terrytao.wordpress.com/2013/02/08/the-harish-chandra-itzykson-zuber-integral-formula/}.

\bibitem{Tao_eigen}
Terrence Tao.
\newblock {\em Eigenvectors from eigenvalues}, 2019.
\newblock
  \url{https://terrytao.wordpress.com/2019/08/13/eigenvectors-from-eigenvalues/}.

\bibitem{Vempala}
Santosh Vempala.
\newblock Personal communication.

\bibitem{JB-minors}
J.-B. Zuber.
\newblock On the minor problem and branching coefficients, 2020.
\newblock \url{https://arxiv.org/abs/2006.03006}.

\end{thebibliography}

\appendix

\section{Disintegration and pushforwards of probability measures} \label{app:disintegration}

The {\it disintegration theorem} is a kind of a factorization result for probability measures.  Given a mapping $\pi:Y \to X$ between two probability spaces satisfying certain mild assumptions, the theorem describes how to ``disintegrate'' a probability measure on $Y$ by decomposing it into a probability measure on $X$ and a family of probability measures supported on the fibers of $\pi$.  Here we merely state the disintegration theorem and describe how it applies in the special case of the Rayleigh map $\mathcal{R} : \mathcal{O}_\Lambda \to GT(\lambda)$.  We refer the reader to \cite{chang1997} for further details.  We also prove two separate lemmas that bound the total variation distance (resp. $\alpha$-divergence) between two probability measures in terms of the total variation distance (resp. $\alpha$-divergence) between their pushforward measures.

\begin{theorem}[Disintegration theorem for probability measures] \label{thm:disintegration}
Let $X$ and $Y$ be complete separable metric spaces equipped with their Borel $\sigma$-algebras, and let $\pi: Y \to X$ be a Borel-measurable function.  Let $\mu$ be a probability measure on $Y$, and write $\pi_* \mu$ for its pushforward by $\pi$.  Then there exists a family of probability measures $\{ \mu_x \}_{x \in X}$ on $Y$, called the disintegrated measures of $\mu$, such that the following hold:
\begin{itemize}
    \item For every Borel set $E \subset Y$, the function $x \mapsto \mu_x(E)$ is Borel-measurable.
    \item The measures $\mu_x$ are supported on the fibers of $\pi$, i.e.~$\mu_x(\pi^{-1}(x)) = 1$ for $\pi_* \mu$-almost all $x \in X$.
    \item For every Borel-measurable function $f: Y \to [0,\infty]$, $$\int_Y f(y)\, d\mu(y) = \int_X \int_{\pi^{-1}(x)} f(y) \, d\mu_x(y) \, d\pi_*\mu(x).$$
\end{itemize}
Moreover, the family of measures $\{ \mu_x \}_{x \in X}$ is $\pi_* \mu$-almost everywhere uniquely determined.
\end{theorem}

\noindent
In the case of the Rayleigh map $\mathcal{R} : \mathcal{O}_\Lambda \to GT(\lambda)$ and the $\U(n)$-invariant probability measure $\mu_\Lambda$ on $\mathcal{O}_\Lambda$, Theorem \ref{thm:disintegration} and Proposition \ref{prop:R-pushfwd} together give the following.

\begin{corollary} \label{cor:mulambda-disint}
For any Borel set $E \subset \mathcal{O}_\Lambda$, any Borel-measurable function $f: \mathcal{O}_\Lambda \to [0, \infty]$, and any Hermitian $Y = \diag(y)$, we have
$$\int_{\mathcal{O}_\Lambda} f(X) \cdot e^{\langle Y, X \rangle } \, d\mu_\Lambda(X) = \frac{1}{\mathrm{Vol}(GT(\lambda))} \int_{GT(\lambda)} \left[\int_{\mathcal{R}^{-1}(P)} f(X) \, d\mathrm{unif}_P(X)\right] \cdot e^{\langle y, \type(P) \rangle } \, dP,$$
where $dP$ is Lebesgue measure on $GT(\lambda)$ and $\mathrm{unif}_P$ is the uniform probability measure on the fiber $\mathcal{R}^{-1}(P)$.
\end{corollary}
\begin{proof}
    This is a corollary of Theorem \ref{thm:equivalence}, since $e^{\langle Y, X \rangle } = e^{\langle y, \mathrm{type}(\mathcal{R}(X))\rangle}$
\end{proof}

\noindent
The following lemma shows that if two probability measures $\mu, \nu$ have the same disintegrated measures, then their total variation distance is equal to the total variation distance between their pushforwards $\pi_* \mu$ and $\pi_* \nu$.  

\begin{lemma} \label{lem:TV-lift}
Let $X, Y$ and $\pi$ be as in Theorem \ref{thm:disintegration}, and let $\mu, \nu$ be probability measures on $Y$ such that $\mu_x = \nu_x$ for $\pi_* \mu$-almost all \textbf{and} $\pi_* \nu$-almost all $x \in X$. Then $\| \mu - \nu \|_{TV} = \| \pi_* \mu - \pi_* \nu \|_{TV}.$
\end{lemma}

\begin{proof}
Recall that $\| \mu - \nu \|_{TV} = \sup_{E \subset Y} | \mu(E) - \nu(E) |,$ where the supremum runs over Borel sets $E \subset Y$.  An equivalent definition is $\| \mu - \nu \|_{TV} = \mathbb{P}(u \ne v)$, where the pair $(u,v) \in Y \times Y$ is distributed according to a maximal coupling of $\mu$ and $\nu$.

We first show $\| \mu - \nu \|_{TV} \ge \| \pi_* \mu - \pi_* \nu \|_{TV}.$ We have:
\begin{align*}
\| \mu - \nu \|_{TV} &= \sup_{E \subset Y} | \mu(E) - \nu(E) | \\
&\ge \sup_{F \subset X} | \mu(\pi^{-1}(F)) -  \nu(\pi^{-1}(F))| = \sup_{F \subset X} | \pi_*\mu(F) - \pi_*\nu(F) | \\& = \| \pi_* \mu - \pi_* \nu \|_{TV},
\end{align*}
where the inequality comes from restricting the supremum to run only over subsets of $Y$ that are preimages of Borel subsets of $X$.

Next we show $\| \mu - \nu \|_{TV} \le \| \pi_* \mu - \pi_* \nu \|_{TV}.$  Let $(a,b) \in X \times X$ be distributed according to a maximal coupling of $\pi_*\mu$ and $\pi_*\nu$.  Let $u \in Y$ be distributed according to $\mu_a$.  If $a = b$, let $v = u$.  Otherwise, let $v$ be distributed according to $\nu_b$.  Then $u \ne v$ if and only if $a \ne b$.  Since $\mu_a = \nu_a$ with probability 1, the joint distribution of the pair $(u, v)$ is a coupling of $\mu$ and $\nu$, and we have:
\[
\| \mu - \nu \|_{TV} \le \mathbb{P}(u \ne v) = \mathbb{P}(a \ne b) = \| \pi_* \mu - \pi_* \nu \|_{TV},
\]
which completes the proof.
\end{proof}

\noindent
An analogous result to Lemma \ref{lem:TV-lift} also holds for the $\alpha$-divergence, defined as follows.

\begin{definition}[$\alpha$-divergence] \label{def:alpha-div}
Let $(X, \Sigma, \beta)$ be a measure space.  Let $\mu, \nu$ be two measures on $X$ that are both absolutely continuous with respect to the reference measure $\beta$, so that we can write $d\mu(x) = p(x) \, d\beta(x)$, $d\nu(x) = q(x) \, d\beta(x)$ for some density functions $p, q$ on $X$. For $0 < \alpha < \infty$ and $\alpha \ne 1$, the $\alpha$-divergence from $\mu$ to $\nu$ is the quantity
\begin{equation} \label{eqn:alpha-div-def}
D_\alpha(\mu \| \nu) = \frac{1}{\alpha - 1} \log \mathbb{E}_\nu \left[ \bigg( \frac{p}{q} \bigg)^\alpha \right] = \frac{1}{\alpha - 1} \log \int_X \frac{p(x)^\alpha}{q(x)^{\alpha - 1}} \, d\beta(x).
\end{equation}
For $\alpha = 0, 1,$ or $\infty$, the $\alpha$-divergence is obtained by taking an appropriate limit in (\ref{eqn:alpha-div-def}), yielding:
\begin{align}
\label{eqn:D0} D_0(\mu \| \nu) &= -\log \mathbb{P}_\nu \big ( p(x) > 0 \big ), \\
\label{eqn:D1} D_1(\mu \| \nu) &= \mathbb{E}_\mu \left[ \log \frac{p}{q} \right], \\
\label{eqn:Dinf} D_\infty(\mu \| \nu) &= \log \operatorname*{ess\,sup}_{x \in X} \frac{p(x)}{q(x)}.
\end{align}
In particular, $D_1(\mu \| \nu)$ coincides with the Kullback--Leibler divergence.
\end{definition}

\noindent
Note that $D_\alpha(\mu \| \nu)$ depends crucially on the choice of reference measure $\beta$, but in most practical settings there is a clear natural choice, such as Lebesgue measure on the real line, the counting measure on a finite set, or the Riemannian volume measure on a manifold.  For fixed $\mu$ and $\nu$, $D_\alpha(\mu \| \nu)$ is a nondecreasing function of $\alpha \in [0, \infty]$ and is continuous on the set of $\alpha$ where it is finite.  When the densities $p$ and $q$ are both continuous, we can replace the essential supremum in (\ref{eqn:Dinf}) with an ordinary supremum as in Definition \ref{def:distr_dist}.

\begin{lemma} \label{lem:alpha-lift}
Let $X, Y$ and $\pi$ be as in Theorem \ref{thm:disintegration}. Suppose that $\mu$ and $\nu$ are both absolutely continuous with respect to some reference measure $\beta$ on $Y$, and suppose that $\mu_x = \nu_x$ for $\pi_* \beta$-almost all $x \in X$.  Then the pushforwards $\pi_* \mu$ and $\pi_* \nu$ are both absolutely continuous with respect to $\pi_* \beta$, and using $\pi_* \beta$ as a reference measure on $X$, we have
\begin{equation} \label{eqn:alpha-lift}
D_\alpha(\pi_* \mu \| \pi_* \nu) = D_\alpha(\mu \| \nu)
\end{equation}
for $0 \le \alpha \le \infty$.
\end{lemma}

\begin{proof}
Write $d\mu(y) = p(y) \, d\beta(y)$, $d\nu(y) = q(y) \, d\beta(y)$.  By Theorem \ref{thm:disintegration}, for any Borel set $E \subset X$, we have
$$ \pi_* \mu(E) = \int_{\pi^{-1}(E)} p(y) \, d\beta(y) = \int_E \int_{\pi^{-1}(x)} p(y) \, d\beta_x(y) \, d\pi_*\beta(x),$$
$$ \pi_* \nu(E) = \int_{\pi^{-1}(E)} q(y) \, d\beta(y) = \int_E \int_{\pi^{-1}(x)} q(y) \, d\beta_x(y) \, d\pi_*\beta(x).$$
Therefore we have $d\pi_* \mu(x) = P(x) \, d\pi_* \beta(x)$, $d\pi_* \nu(x) = Q(x) \, d\pi_* \beta(x)$, where
$$P(x) = \int_{\pi^{-1}(x)} p(y) \, d\beta_x(y), \qquad Q(x) = \int_{\pi^{-1}(x)} q(y) \, d\beta_x(y).$$
Again applying Theorem \ref{thm:disintegration} and using our assumption that $\mu_x = \nu_x$ almost everywhere, we find that for any Borel set $F \subset Y$ we have
\begin{align*}
\mu(F) = \int_Y \mathbbm{1}_F(y) \, d\mu(y) &= \int_X  \int_{\pi^{-1}(x)} \mathbbm{1}_F(y) \, d\mu_x(y) \, d\pi_* \mu(x) \\
&= \int_X \left[\int_{\pi^{-1}(x)} \mathbbm{1}_F(y) \, d\mu_x(y)\right] P(x)\, d\pi_* \beta(x) \\
&= \int_X \left[\int_{\pi^{-1}(x)} \mathbbm{1}_F(y) \, d\nu_x(y)\right] \frac{P(x)}{Q(x)} \cdot Q(x) \, d\pi_* \beta(x) \\
&= \int_X \left[\int_{\pi^{-1}(x)} \mathbbm{1}_F(y) \cdot \frac{P(\pi(y))}{Q(\pi(y))} \, d\nu_x(y)\right] Q(x) \, d\pi_* \beta(x) \\
&= \int_F \frac{P(\pi(y))}{Q(\pi(y))} \, d\nu(y).
\end{align*}
Additionally, we have
\[
    \mu(F) = \int_Y \mathbbm{1}_F(y) \, d\mu(y) = \int_F p(y) \, d\beta(y) = \int_F \frac{p(y)}{q(y)} \cdot q(y) \, d\beta(y) = \int_F \frac{p(y)}{q(y)} \, d\nu(y).
\]
Since these two expressions for $\mu(F)$ hold for any Borel set $F \subset Y$, combining them then implies
\begin{equation} \label{eqn:pq-ratio}
\frac{p(y)}{q(y)} = \frac{P(\pi(y))}{Q(\pi(y))}
\end{equation}
for $\beta$-almost all $y \in Y$.

For $\alpha \in (0, \infty)$ and $\alpha \ne 1$, we then have
\begin{align*}
D_\alpha( \mu \| \nu) &= \frac{1}{\alpha - 1} \log \mathbb{E}_\nu \left[ \bigg( \frac{p}{q} \bigg)^\alpha \right] \\
&= \frac{1}{\alpha - 1} \log \mathbb{E}_\nu \left[ \bigg( \frac{P \circ \pi}{Q \circ \pi} \bigg)^\alpha \right] \\
&= \frac{1}{\alpha - 1} \log \mathbb{E}_{\pi_* \nu} \left[ \bigg( \frac{P}{Q} \bigg)^\alpha \right] = D_\alpha( \pi_*\mu \| \pi_*\nu),
 \end{align*}
 as desired.  Identical statements for $\alpha = 0, 1$ and $\infty$ follow by applying (\ref{eqn:pq-ratio}) to (\ref{eqn:D0}), (\ref{eqn:D1}) and (\ref{eqn:Dinf}) respectively:
 \begin{align*}
     D_0(\mu \| \nu) &= -\log \mathbb{P}_\nu \big( p(y) > 0 \big) = -\log \mathbb{P}_{\pi_*\nu} \big( P(x) > 0 \big) = D_0(\pi_* \mu \| \pi_*\nu), \\
     D_1(\mu \| \nu) &= \mathbb{E}_\mu \left[ \log \frac{p}{q} \right] = \mathbb{E}_\mu \left[ \log \frac{P \circ \pi}{Q \circ \pi} \right] = \mathbb{E}_{\pi_* \mu} \left[ \log \frac{P}{Q} \right] = D_1(\pi_* \mu \| \pi_* \nu), \\
     D_\infty(\mu \| \nu) &= \log \operatorname*{ess\,sup}_{y \in Y} \frac{p(y)}{q(y)} = \log \operatorname*{ess\,sup}_{x \in X} \frac{P(x)}{Q(x)} = D_\infty(\pi_* \mu \| \pi_* \nu).
 \end{align*}
\end{proof}

\end{document}